\documentclass[11pt,english]{article}
\usepackage[T1]{fontenc}
\usepackage[latin9]{inputenc}
\usepackage{geometry}
\usepackage{units}
\usepackage{algorithm2e}
\usepackage{amsmath}
\usepackage{amsthm}
\usepackage{amssymb}
\usepackage{graphicx}
\usepackage{setspace}
\usepackage[authoryear]{natbib}
\makeatletter
\theoremstyle{plain}
\newtheorem{lem}{\protect\lemmaname}
\theoremstyle{plain}
\newtheorem{thm}{\protect\theoremname}
\theoremstyle{plain}
\newtheorem{cor}{\protect\corollaryname}
\theoremstyle{definition}
 \newtheorem{example}{\protect\examplename}

\usepackage{babel}

\makeatother

\usepackage{babel}
\providecommand{\corollaryname}{Corollary}
\providecommand{\examplename}{Example}
\providecommand{\lemmaname}{Lemma}
\providecommand{\theoremname}{Theorem}

\begin{document}
\title{Distributionally Robust Optimal Auction Design under Mean Constraints\thanks{I am grateful to Ben Brooks, Songzi Du, Vitor Farinha Luz, Jonathan
Libgober, Eran Shmaya, Alex Wolitzky, Doron Ravid and Gabriel Carroll
for their helpful comments. I am particularly thankful to Chris Ryan
and Yeon-Koo Che, for their support and generosity with their time.}}
\author{Ethan W. Che\thanks{Graduate School of Business, Columbia University.}}
\date{First Draft: 11/16/2019. Current Draft: 02/15/2022}
\maketitle

% Abstract. Note that this must come before \maketitle.
\begin{abstract}
We study a seller who sells a single good to multiple bidders with
uncertainty over the joint distribution of bidders' valuations, as
well as bidders' higher-order beliefs about their opponents. The seller
only knows the (possibly asymmetric) means of the marginal distributions
of each bidder's valuation and the range. An adversarial nature chooses
the worst-case distribution within this ambiguity set along with the
worst-case information structure. We find that a second-price auction
with a symmetric, random reserve price obtains the optimal revenue
guarantee within a broad class of mechanisms we refer to as \emph{competitive
}mechanisms, which include standard auction formats, including the
first-price auction, with or without reserve prices. The optimal mechanism
possesses two notable characteristics. First, the mechanism treats
all bidders identically even in the presence of ex-ante asymmetries.
Second, when bidders have identical means and the number of bidders $n$ grows large, the seller's optimal
reserve price converges in probability to a non-binding reserve price
and the revenue guarantee converges to the best possible revenue guarantee
at rate $O(1/n)$.



\end{abstract}

% Paper body
\section{Introduction}

In the standard auction theory, beginning with \citet{Myerson1981},
the seller's uncertainty over the bidders' private values is typically
modeled as a well-defined prior distribution. While this framework
is flexible, and allows for a wide range of beliefs, it does not describe
how such beliefs are formed in the first place or what happens if
these beliefs fail to approximate reality.  Indeed, the seller not
only faces uncertainty over the realizations of a probability distribution,
but also the deeper uncertainty that they may not have the correct
model at all.

How does a seller, who is concerned about model misspecification,
choose a selling mechanism? Perhaps it is unrealistic to assume that
a seller begins with a full distribution over types, which is high-dimensional
and assigns a specific probability to all possible circumstances.
It is more realistic that the seller carries low-dimensional ``summary
statistics'' or ``rules of thumb,'' such as the mean, variance,
mode, etc., on which he bases his decisions. One reasonable criterion
for choosing a mechanism, if robustness is desired, is to find the
mechanism that performs best under the worst-case scenario with respect
to all distributions that match these moments and all possible higher
order beliefs consistent with the moments.

We consider a setting in which a seller sells a single good to one
of $n$ bidders. The seller lacks knowledge of the joint distribution
of valuations, including the correlation structure between bidders'
valuations and their marginal distributions. Instead, the seller knows
the range of the bidders' valuations and the (possibly asymmetric)
means of each bidder's marginal distribution. While the results in
this paper can be extended to include knowledge of higher moments
of the distribution, we focus on the mean and range constraints, as
often in practice these are known with the highest degree of confidence 
and there tends to be greater uncertainty over higher moments.

In addition, the seller knows that the bidders' valuations are private
and bidders always play undominated strategies, but he is uncertain
about the information bidders have about their opponents' values.
In this sense, we are interested not only in finding mechanisms which
are robust to bidders' first-order valuations, but also robust to
their higher order beliefs about one another.

The objective of this paper is to find the selling mechanism that
obtains the \emph{optimal revenue guarantee} in this setting, the
highest revenue the seller is guaranteed to achieve regardless of
the joint distribution and possible information bidders have about
their opponents. To solve for the revenue guarantee, we use an equivalent
characterization of the guarantee as the equilibrium revenue of a
simultaneous-move, zero-sum game between the seller and an adversarial
nature. Nature chooses the joint distribution of valuations, consistent
with the mean and range constraints, as well as the bidders' information
structures, in order to minimize the seller's revenue given the mechanism.
The seller chooses the mechanism to maximize his revenue given the
worst-case joint distribution and beliefs. The min-max theorem implies
that the equilibrium revenue of the zero-sum game gives the optimal
revenue guarantee for the seller.

We first find the robustly optimal reserve price for the second-price
auction. We then show that the second-price auction with the optimal
reserve price achieves the optimal revenue guarantee within a wide
class of mechanisms we call \emph{competitive} mechanisms. Competitive
mechanisms comprise all mechanisms that never assign the good to any
bidder who draws a ``clearly low'' value.  Intuitively, whenever
there are two bidders who know with complete certainty among themselves that one of
them has a lower value than the other, a competitive mechanism
should never allocate the good to the bidder with the lower value
(which we denote as a ``clearly low'' value). 

This notion can be seen as a natural generalization of efficiency.
The class of efficient mechanisms is often too restrictive as it requires
that the highest value bidder must win in \emph{all} realized value
profiles under \emph{all} information structures, and there are many
common auction equilibria in which an efficient allocation does not
arise. Unlike efficient mechanisms, competitive mechanisms only forbid
allocation to low value bidders if there is \emph{zero} uncertainty
that they have a low value given common knowledge of the value distribution.
For example, if the bidders' values are drawn iid from a non-degenerate distribution
$\mu$, even if a bidder draws the lowest value in the support of
$\mu$ they may still be allocated the good by a competitive mechanism,
since his opponents cannot determine which value he drew. Since the
restriction for competitive mechanisms only concerns the behavior
of the mechanism under a particular information structure within a
small set of realized value profiles, the definition permits a wide
range of mechanisms including many standard auction formats, such
as the first-price auction. We remark that there are many mechanisms
outside of this class that are revenue-guarantee dominated by the
second-price auction with the optimal reserve price. We highlight
this class of mechanisms due to its wide scope and intuitive appeal.

There are several notable features of the optimal mechanism. First,
the mechanism is symmetric across all bidders, even if the bidders
differ in their expected valuations of the good. Second, we find that
as the number of bidders $n$ (with symmetric means) grows large,
the optimal reserve price converges in probability to a non-binding
-- effectively zero -- reserve price. The worst-case expected revenue
converges to the mean, i.e. the best possible revenue guarantee, at
rate $O(1/n)$. 

We find that the worst-case information structure is the uninformed
common-prior information structure, which gives the bidders common
knowledge of the value distribution but no other information about
their opponents (it is interesting that this is precisely the information structure textbook auction theory assumes).
The worst case value distribution randomly picks
one bidder to be the highest bidder and suppresses the values of remaining
bidders to a fixed lower bound $\alpha$. This captures the intuition that
the worst-case scenario for the seller involves suppressed competition
between the bidders. Our results show that the second-price auction
performs well even in this scenario and nature's ability to suppress
competition diminishes as the number of bidders grows large.

 It is striking that a simple standard auction format that treats
bidders symmetrically emerges as a robustly optimal, among a very
large class that may include potentially very complex mechanisms and
those that treat bidders differently, e.g., discriminating one against
another. This result can be seen as an optimality certificate for
a simple, symmetric, and easily-implementable mechanism, which requires
minimal knowledge of the true distribution. The robust optimality
of the second-price auction -- or its more practical equivalent,
the English auction -- accords well with its prevalence in practice,
a result similar in spirit to \citet{Myerson1981}. However, in contrast
to \citet{Myerson1981}, the optimal mechanism treats all bidders
identically even when the seller's a priori knowledge about the bidders
is asymmetric, and the reserve price becomes negligible as the number
of bidders increases.

This paper joins the growing literature of robust mechanism design,
beginning with \citet{scarf1957minmax}. There is much work that addresses
robust design in auctions, such as \citet{bbm2016informationally},
\citet{brooksdu2019}, and \citet{du2018commonvalue}.\footnote{See also \citet{bergemannschlag2008pricing}, \citet{wolitzky2016bilateral},
\citet{carroll2015linearcontracts,carroll2017robustseparation,carroll2019infoscore},
\citet{mu2019intertemporal}, and \citet{chenli2018revising} for
work on robust mechanism design in other environments or with alternative
specifications.} The papers closest to ours are: \citet{He2019RobustlyOptimal}, \citet{bbm2019revenueguarantee},
\citet{distributionallyRobust2019}, \citet{suzdaltsev2020}, \citet{carrasco2018optimalselling},
and \citet{brooks-du-2021}.

\citet{carrasco2018optimalselling} study the problem of a seller
faced with a \emph{single} buyer from an unknown distribution subject
to an arbitrary number of moment constraints, and solve the problem
with a known mean and range as a special case. This paper can be seen
as a generalization of this model to $n\geq2$ potential buyers. While
the method based on duality is similar, the application of the method
as a verification tool becomes more demanding. While their method
requires verifying the value distribution in a single dimensional
case, the current setup with multiple bidders requires us to conjecture
the full-fledged joint distribution with nontrivial correlation among
bidders' valuation, as will later become clear. Furthermore, our analysis,
by varying the number of bidders, allows us to study the impact of
competition. More importantly, the choice of a selling mechanism (i.e.,
the auction format) has no analogue in that paper.

\citet{suzdaltsev2020} finds the optimal mechanism among deterministic,
dominant strategy incentive compatible (DSIC) and ex post individually
rational mechanisms given the same constraints on the joint distribution
considered in this paper. He finds that the optimal mechanism is a
linear score auction, in which scores are calculated as bidder-specific
linear functions of the bids. In contrast, we consider competitive,
Bayesian mechanisms which allows us to prove optimality in comparison
to commonly-used auction rules which are not deterministic nor DSIC,
such as the first-price auction. In addition, his mechanism is only
robust to bidders' first-order valuations, whereas the featured mechanism
in this work is also robust to bidders' higher order beliefs. 

\citet{distributionallyRobust2019} studies the Stackelberg version
of our robust second-price auction problem, in which nature has a
second-mover advantage and fully knows the reserve price before choosing
the worst-case distribution subject to the same mean and range constraints.
They also find an optimal highest-bidder lottery mechanism, which
implements the same outcome as our second-price auction under nature's
worst-case distribution,\footnote{Strictly speaking, the optimal highest-bidder lottery differs from
a second-price auction with a random reserve price, as it involves
a different allocation rule. The difference becomes irrelevant at
the optimal distribution chosen by nature, as it puts zero probability
to the types where the allocation rules differ.} and they prove the optimality of this mechanism within a smaller
class of mechanisms -- ex-post incentive compatible mechanisms in
which only the highest value bidder is allocated the good -- which
does not include many common auction rules such as the first-price
auction. Also, they only consider an environment where the bidders
are symmetric, unlike in this work which allows for asymmetric mean
constraints.

\citet{He2019RobustlyOptimal} study the second-price auction within
a setting in which a seller is uncertain over the joint distribution
of valuations but \emph{knows the exact marginal distributions}. Nature
in our model can be seen as not only choosing the correlations between
bidders' valuations, but also the marginal\emph{ }distribution for
each bidder. This additional choice matters. As will be seen later,
the optimally chosen marginal distribution involves both mass points
and a smooth density, which is outside the set of marginal distributions
they consider.\footnote{They assume that the exogenously given marginal distribution $F$
admits positive density $f$ everywhere. They further require that
$xf(x)$ is weakly increasing in $x$. These conditions fail for our
optimally chosen marginal distribution.} Further, the optimal marginal distribution changes with the number
of bidders, which they do not allow for. 
Even though the seller faces greater
uncertainty in our setting, our mechanism maintains the same $O(1/n)$ rate of convergence to the best 
possible revenue guarantee as $n$ grows large. Most importantly, the auction
rule is fixed in their paper to be the second-price auction, whereas
the second-price auction is shown to be optimal within a large class
of mechanisms in the current paper.  \cite{bei-2019} studies a similar setting and find
that a sequential posted-price mechanism achieves a 4.78 approximation
to the optimal DSIC mechanism. When bidders have identical marginals they find
that asymptotically, the mechanism obtains
the optimal worst-case revenue and achieves an $n/(n-1)$ approximation among
all second-price auctions with a common reserve price. As mentioned before,
in our model, the seller is also uncertain of the marginal distributions themselves,
and we show optimality of the second-price auction within a large class of mechanisms that
includes non-DSIC mechanisms in this setting.

\citet{bbm2019revenueguarantee} perform a similar exercise of comparing
standard auction mechanisms according to their revenue guarantee.
They study the interdependent value environment in which there is
a $\emph{commonly known}$ value distribution which is symmetric across
bidders, but the seller is uncertain about the signals possessed by
the bidders. They find revenue guarantee optimality of the first-price
auction among standard auctions, and revenue equivalence among standard
auctions under the restricted domain of common value with symmetric
affiliated signals. Our model does not assume a common knowledge of
value distribution, and we allow for asymmetric means and a broader
class of mechanisms. At the same time, we do not allow for interdependent
values. In this sense, the current paper complements this work.

\citet{brooks-du-2021} find a minimax mechanism subject to an adversarial
nature who chooses the information structure, value distribution,
and the equilibrium in order to minimize the seller's revenue, given
the same mean and range constraints as discussed in this paper. The
minimax mechanism is a proportional auction rule and the worst-case
information structure involves independent and exponentially-distributed
signals for each bidder. The current paper complements their work
as it can be seen as a constrained version of their problem in which
nature can choose only among private-value information structures.
This constraint makes the analysis of the present paper distinct from
theirs in a nontrivial way, and their method and analysis do not apply
to the current paper. In this sense, the current paper addresses auction
settings in which it is natural to assume that bidders know their
own valuations and it is of greater concern for the seller to address
the uncertainty about the bidders' valuations or in the information
they possess about each other. It is also noteworthy and of appeal
that the current paper finds a standard auction format as the
optimal mechanism within a wide range of mechanisms that include other
commonly used auction rules. Finally, the optimal mechanism in this
paper is symmetric across all bidders even if bidders are known to
be asymmetric ex-ante, which is in contrast with the proportional
auction rule they propose, which assigns different allocation probabilities
to different bidders and even assigns zero allocation probability
for bidders whose means are too low.

Finally, the current paper is related to the literature on discriminatory
auction design. Following \citet{Myerson1981}, the authors have put
forth arguments that treating (some times even symmetric) bidders
asymmetrically may increase revenue (e.g. \citet{ayres-cramton-1996}).
Some authors argue that the symmetric treatment of bidders may result
from legal constraints mandating nondiscrimination. \citet{deb-pai-2017}
suggested ways to achieve symmetry with a more complex auction rule.
By contrast, the current paper suggests that symmetry can arise naturally
if the seller is concerned about robustness. To the best of our knowledge,
the current paper is the first to make this observation.

The remainder of the paper is organized as follows. Section \ref{sec:Model}
introduces the model of the paper. In section \ref{sec:second-price-auction},
we fix the mechanism to be a second-price auction and obtain the optimal
reserve price. In section \ref{sec:Optimal-Competitive-Mechanism},
we introduce the class of competitive mechanisms and prove that the
second-price auction with the optimal reserve price achieves the highest
revenue guarantee among all competitive mechanisms. Section \ref{sec:Conclusion}
concludes.

\section{\label{sec:Model}Model}

Consider a seller selling a single good to one of $n$ bidders. The
valuations of the bidders are distributed according to a distribution
$F$, which is unknown to the seller. Instead, the seller knows that
the support of $F$ is contained in $V:=[0,1]^{n}$ and knows the
means for each of bidders' marginal distributions $\{m_{i}\}_{i=1}^{n}$,
where $m_{i}\in(0,1)$ for each bidder $i$. Without loss, we index
the bidders and the means in descending order of the means (unless
specified otherwise): $m_{1}\geq m_{2}\geq...\geq m_{n}$. We denote
the set of probability distributions satisfying these constraints
as $\mathcal{P}\left(V,\{m_{i}\}_{i=1}^{n}\right)$, and in general
we denote the set of probability distributions on an underlying set
$X$ as $\mathcal{P}(X)$.

The seller also knows that each bidder $i$ observes his value $v_{i}$,
which makes this a private value auction. However, he is uncertain
about other bidders' values or the information they may have about their opponents. We denote
$I\in\mathcal{I}$ to be the information structure that defines the
information possessed by the bidders. Formally, we consider an admissable
class of information structures $\mathcal{I}$ such that for all $I\in\mathcal{I}$,
$I\equiv(S,K)$ where $S\equiv\prod_{i=1}^{n}(\Sigma_{i},V_{i})=\Sigma\times V$
is a measurable set of signals (including the bidder's own valuation)
and $K\in\mathcal{P}(\Sigma|V)$ is a conditional distribution over
$\Sigma$ given realized valuations. In other words, each bidder
observes $(\sigma_{i},v_{i})\in(\Sigma_{i},V_{i})$, where $v_{i}$
is the bidder's valuation and $\sigma_{i}$ is a signal. We impose
no restrictions on the signal space $\Sigma_{i}$, which is allowed
to be very general. For instance, it may contain one's high-order
beliefs about their opponents' values and their higher-order beliefs,
etc. It is also possible that bidders may only have partial or no knowledge about
the value distribution. One important special case is the \emph{uninformed
common-prior} (UCP) information structure $I_{0}\in\mathcal{I}$,
in which each bidder observe the signal $\sigma_{i}=F$ along with
their private values, i.e., the bidders have (1) common knowledge
of the value distribution $F$ and (2) no other knowledge of their
opponents valuations.

The seller chooses a mechanism $M=(A,x,t)$, where $A=\prod_{i=1}^{n}A_{i}$
is the set of actions for each player, $x:A\to[0,1]^{n}$ is the allocation
rule subject to the constraint $\sum_{i=1}^{n}x_{i}(a)\leq1$, and
$t:A\to\mathbb{R}^{n}$ is the transfer rule. The action set is also
allowed to be very general, and can even involve bidders reporting
their own types and beliefs (and possibly all higher-order beliefs)
\citep[see][]{bergemann-morris-2005}. A mechanism is a direct mechanism
if $A=V$. A value distribution $F$, information structure $I$ and
a mechanism $M$ define a Bayesian game $\Gamma(M,I,F)$ among the
bidders. Bidders' signals $S=(\Sigma,V)$ are drawn according to $Q\equiv K\times F\in\mathcal{P}(S)$.
Each bidder plays a strategy $\beta_{i}:(S_{i},V_{i})\to\mathcal{P}(A_{i})$,
which determines a probability distribution over actions given the
realized signal. Once the actions are realized, the allocation and
transfers are determined according to $(x,t)$. 

A strategy profile $\beta=\{\beta_{i}\}_{i=1}^{n}$ is a Bayes-Nash
equilibrium if the individual rationality (IR) and incentive-compatibility
conditions (IC) are met for all $i\in[n]$:

\begin{align}
U_{i}(\beta) & \geq0,\tag{IR}\label{eq:(IR)-1}\\
U_{i}(\beta) & \geq U_{i}(\beta_{i}',\beta_{-i}),\;\;\forall\beta_{i}',\tag{IC}
\end{align}
where 
\[
U_{i}(\beta)\equiv\int_{S}\int_{A}\left[v_{i}x_{i}(a)-t_{i}(a)\right]\beta(da|s)Q(ds).
\]
Under the UCP information structure $I_{0}$ and any fixed value
distribution $F$, the revelation principle applies, so one can without
loss consider a direct mechanism where $A=V$ that satisfies the Bayesian
individual-rationality (IR-0) and incentive-compatibility (IC-0) constraints:

\begin{align}
U_{i}(v_{i}) & \geq0,\tag{IR-0}\label{eq:IR-0}\\
U_{i}(v_{i}) & \geq U_{i}(\hat{v}_{i}|v_{i}),\;\;\forall\hat{v}_{i},\tag{IC-0}\label{eq:IC-0}
\end{align}
in which $U_{i}(\hat{v}_{i}|v_{i})\equiv\mathbb{E}_{v_{-i}\sim F}[u_{i}(\hat{v}_{i},v_{-i}|v_{i})]$
is the interim utility of bidder $i$ when he reports value $\hat{v}_{i}$,
$U_{i}(v_{i})\equiv U_{i}(v_{i}|v_{i})$ and $u_{i}(v|v_{i})\equiv v_{i}x_{i}(v)-t_{i}(v)$
is the realized utility.

Given a game $\Gamma(M,I,F)$, we denote the set of Bayes-Nash equilibria
in undominated strategies as $B(M,I,F)$. Under a Bayes-Nash equilibrium,
the revenue is simply the sum of the expected transfers from each
bidder:

\[
\mathcal{R}(M,I,F,\beta)\equiv\sum_{i=1}^{n}\int_{S}\int_{A}t_{i}(a)\beta(da|s)Q(ds)
\]

Let $\mathcal{R}(M,I,F)=\inf_{\beta\in B(M,I,F)}\mathcal{R}(M,I,F,\beta)$
be the worst-case revenue across all Bayes-Nash equilibria in undominated
strategies. If $B(M,I,F)$ is empty, we let $\mathcal{R}(M,I,F)=0$.
 We say that mechanism $M$ \emph{guarantees} revenue $R$ if $R\leq\mathcal{R}(M,I,F)$
for all $F\in\mathcal{P}\left(V,\{m_{i}\}_{i=1}^{n}\right),I\in\mathcal{I}$.
We say that $M^{*}$ is \emph{revenue guarantee optimal} within a class of
mechanisms $\mathcal{M}$ if $M^{*}\in\arg\max_{M\in\mathcal{M}}\min_{F,I}\mathcal{R}(M,I,F).$
In order to solve for the revenue guarantee optimal mechanism, we
consider a simultaneous-move zero-sum game between an adversarial
nature and the seller. Nature chooses the worst case distribution
and the information structure $\left(F,I\right)$ in order to minimize
the seller's expected revenue given the seller's mechanism $M$, and
the seller chooses $M$ in order to maximize his utility given the
worst-case distribution $\left(F,I\right)$.

We define an equilibrium of this game to be a value distribution and
information structure $\left(F^{*},I^{*}\right)$ and mechanism $M^{*}\in\widehat{\mathcal{M}}$
within some admissible class $\widehat{\mathcal{M}}$, which we shall
specify in Section \ref{sec:Optimal-Competitive-Mechanism}, such
that:

\begin{equation}
\mathcal{R}\left(M,I^{*},F^{*}\right)\leq\mathcal{R}\left(M^{*},I^{*},F^{*}\right)\leq\mathcal{R}\left(M^{*},I,F\right),\label{eq:mechanism_saddle_point}
\end{equation}
for all $F\in\mathcal{P}\left(V,\{m_{i}\}_{i=1}^{n}\right)$, information
structures $I\in\mathcal{I}$, and mechanisms $M\in\widehat{\mathcal{M}}$.

By the min-max theorem, $M^{*}$ is revenue guarantee optimal if $(M^{*},I^{*},F^{*})$
satisfies (\ref{eq:mechanism_saddle_point}):\footnote{By the max-min inequality, $\sup_{M}\inf_{F,I}\mathcal{R}(M,I,F)\leq\inf_{F,I}\sup_{M}\mathcal{R}(M,I,F)$.
From the saddle point condition, we can obtain the reverse inequality:
\[
\inf_{F,I}\sup_{M}\mathcal{R}(M,I,F)\leq\sup_{M}\mathcal{R}(M,I^{*},F^{*})\leq\mathcal{R}(M^{*},I^{*},F^{*})\leq\inf_{F,I}\mathcal{R}(M^{*},I,F)\leq\sup_{M}\inf_{F,I}\mathcal{R}(M,I,F).
\]
Combining these two inequalities, we obtain the desired equality:
$\mathcal{R}(M^{*}, I^{*},F^{*})=\sup_{M}\inf_{F}\mathcal{R}(M,I,F)=\inf_{F}\sup_{M}\mathcal{R}(M,I,F).$
\citep[see Proposition 22,][]{osborne-rubin}}

\[
\mathcal{R}(M^{*},I^{*},F^{*})=\sup_{M}\inf_{F,I}\mathcal{R}(M,I,F)=\inf_{F,I}\sup_{M}\mathcal{R}(M,I,F).
\]
This max-min characterization implies that all equilibria will be
payoff equivalent, even though there could be multiple equilibria
that generate this revenue.

We proceed our analysis in two steps. In Section \ref{sec:second-price-auction},
we fix the auction rule to be a second-price auction and allow the
seller to choose the reserve price optimally against the worst case distribution
$F^{*}$ chosen by nature. We prove the saddle-point inequalities
among this class of mechanisms, and illustrate some properties of
the equilibrium. In Section \ref{sec:Optimal-Competitive-Mechanism},
we prove that the optimal second-price auction obtains the optimal
revenue guarantee among a wide class of mechanisms we refer to as
\emph{competitive} mechanisms (which we will formally define later).

\section{\label{sec:second-price-auction}Robustly Optimal Second-price Auction}

In this section, we restrict our attention to second-price auction
mechanisms and solve for the optimal reserve price within this class
of mechanisms. We also find the worst-case distribution, with which
the optimal reserve price obtains the revenue guarantee.

\subsection{Second-price auctions with random reserve price}

Suppose the seller uses a second-price auction and sets his reserve
price $p$ according to some CDF $G$. Let $S_{G}$ denote the resulting
mechanism. Since a bidder with valuation $v_{i}$ has a dominant strategy
of bidding his valuation $v_{i}$, regardless of his information about
his opponents, he will win if his valuation is highest and exceeds
the reserve price. Assuming that ties are broken at random, $S_{G}$
can be expressed in direct mechanism $(V,x,t)$ where: 

\begin{align*}
x_{i}({\bf v}) & =G(v_{i})\frac{1_{\left\{ v_{i}=v^{(1)}\right\} }}{\#\{v_{j}=v^{(1)}\}}\\
t_{i}({\bf v}) & =\mathbb{E}_{G}\left[\max\{v^{(2)},p\}\cdot1_{\{v^{(1)}>p\}}\right]\cdot1_{\left\{ v_{i}=v^{(1)}\right\} },
\end{align*}
where $v^{(k)}$ denotes the $k$th highest valuation among $(v_{1},...,v_{n})$,
and $\mathbb{E}_{G}$ refers to expectation with respect
to $p\sim G$.

Since the mechanism is fixed at second-price auctions, except for
the reserve price, it remains to identify its distribution $G$ that
is robustly optimal. As mentioned in the general setting, the seller's
belief over valuations and his reserve price strategy is determined
as the outcome of zero-sum, simultaneous move game between him and
an adversarial nature. The seller chooses a (possibly degenerate)
reserve price $G\in\mathcal{P}\left([0,1]\right)$ given the worst-case
distribution $F$. Nature chooses the distribution of bidders' valuations
$F\in\mathcal{P}\left(V,\{m_{i}\}_{i=1}^{n}\right)$ and the bidders'
information structure $I\in\mathcal{I}$ given the seller's reserve
price $G$. 

However, unlike in the setting where the seller picks an arbitrary
mechanism $M$, nature's choice of the information structure $I$
does not affect the seller's revenue, since bidders will play their
dominant strategies regardless of the realized signals. In other words,
for any fixed value distribution $F$ and reserve price $G$, there
is a unique truth-telling Bayes-Nash equilibrium in undominated strategies
$\beta_{i}(v_{i}),\forall v_{i},\forall i$, which is the same across
all private-value information structures. This implies, 

\[
\mathcal{R}(S_{G},I,F)=\inf_{\beta\in B(F,I,M)}\mathcal{R}(S_{G},I,F,\beta)=\inf_{\beta\in B(F,I',M)}\mathcal{R}(S_{G},I',F,\beta)=\mathcal{R}(S_{G},I',F,),\text{ }\forall F,I,I',G.
\]
Because of this, we will ignore the information structure and focus
solely on the value distribution. Throughout this section, we will
assume that $F$ is equipped with the uninformed common-prior (UCP) information
structure $I_{0}$.

Given this behavior, we denote the expected revenue given $F$ and
$G$, as follows:

\[
\Psi(G,F):=\mathcal{R}(S_{G},I_{0},F)=\sum_{i=1}^{m}\int_{V}t_{i}({\bf v})dF({\bf v})=\int_{V}\int_{0}^{1}\max\{v^{(2)},p\}\cdot1_{\{v^{(1)}>p\}}dG(p)dF({\bf v}).
\]

Since the auction format is fixed and the information structure is
irrelevant (and set to $I_{0}$ without loss), an equilibrium of this
restricted game between nature and the seller is defined to be a pair
of strategies $G^{*}\in\mathcal{P}\left([0,1]\right)$ and $F^{*}\in\mathcal{P}\left(V,\{m_{i}\}_{i=1}^{n}\right)$
such that:

\begin{equation}
\Psi(G,F^{*})\leq\Psi(G^{*},F^{*})\leq\Psi(G^{*},F),\label{eq:saddle_point}
\end{equation}
for all $F\in\mathcal{P}\left(V,\{m_{i}\}_{i=1}^{n}\right)$ and $G\in\mathcal{P}\left([0,1]\right)$.
The equilibrium revenue can be interpreted as the optimal revenue
guarantee among reserve price strategies:

\[
\Psi(G^{*},F^{*})=\sup_{G}\inf_{F}\Psi(G^{*},F^{*})=\inf_{F}\sup_{G}\Psi(G^{*},F^{*}).
\]

Toward finding the revenue guarantee, our strategy will be to present
a conjectured strategy profile and verify that it satisfies the saddle-point
inequalities. Before presenting the main result of this paper, which
provides the solution with $n\geq2$ bidders, it is helpful to first
illustrate the solution in a simpler setting.

\subsection{Two bidder example}

We illustrate our results with the case of two bidders with symmetric
mean $m\in(0,1)$. First, observe that as in the single buyer case
in \citet{carrasco2018optimalselling}, there is no pure strategy
equilibrium. Assume, for the sake of argument, that the marginal mean
is $m=1/2$.\footnote{This is generally the case for any $n\geq2$ bidders and for any $m\in(0,1)$
as can be seen in the Stackelberg second-price auction equilibrium
in \citet{distributionallyRobust2019}.} Suppose, to the contrary, there were a pure strategy equilibrium
in which the seller chooses a (deterministic) reserve price $p$.
Clearly $p<m$, or else nature will put all mass on the value profile
$(m,m)$, resulting in no sale with probability one. For any $p<m$ the seller may choose,
nature's best response will put positive mass at $(p,p)$, in which
case no sale occurs and the seller receives zero revenue (assuming
ties are broken in favor of nature). However, if the seller lowers
his reserve price to some $p'<p$, the reserve price is no longer
binding at $(p,p)$ and a sale occurs with revenue equal to $p$.
Since the deviation is strictly profitable, there is no pure strategy
equilibrium. Therefore the seller must randomize in reserve price
in equilibrium.

\begin{figure}[h]
\caption{\label{fig:hyperplane_support} $\phi(\mathbf{v};G^{*})$ and the
support of $F^{*}$ in the two bidder equilibrium ($m=\frac{1}{2})$.}

\bigskip{}

\begin{centering}
\includegraphics[width=5in]{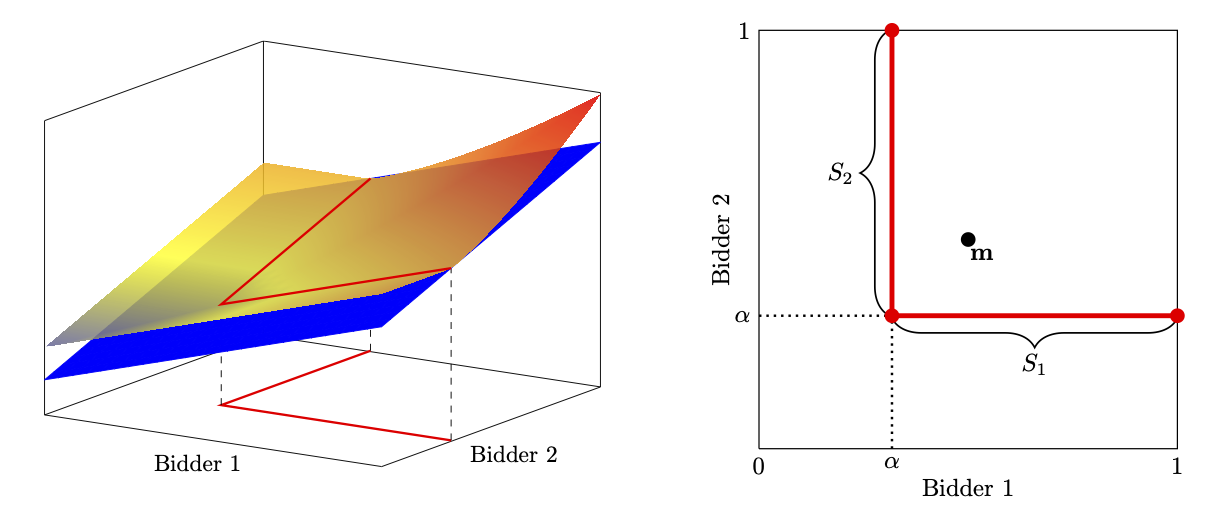}
\par\end{centering}
\bigskip{}

{\scriptsize{}{}{{\scriptsize{}{}Notes:}} On the left are
the expected revenue function $\phi({\bf v};G^{*})$ (in red) induced
by the seller's equilibrium strategy $G^{*}$, the supporting affine
function $L(\mathbf{v})$ (in blue), and the intersection of the two
(in bold red). On the right, the support of $F^{*}$ (in bold red)
and the mean point. When $m=\frac{1}{2}$, $\alpha\approx0.317$.}{\scriptsize\par}
\end{figure}

To find an equilibrium with random reserve price, we first conjecture
both nature's and the seller's strategies, and verify that they constitute
a saddle point. We conjecture that nature's strategy $F^{*}$ randomly
picks one bidder to be the highest bidder, with value randomly drawn
from CDF $H$, and sets the value of the second highest bidder to
be as low as possible, perhaps to some lower bound $\alpha$ (which
is yet to be determined), producing value profiles of the form $(v,\alpha)$
or $(\alpha,v)$, where $v$ is drawn from $H(v)$. Visually, the
support of $F$ has an ``L''-shape, as can be seen in the right
panel of Figure \ref{fig:hyperplane_support}. Formally, $\text{supp}(F^{*})=S_{1}\cup S_{2}$
where $S_{1}:=\left\{ (v,\alpha):v\in(\alpha,1]\right\} $ and $S_{2}:=\left\{ (\alpha,v):v\in(\alpha,1]\right\} $.
Each segment $S_{i}$ contains the realized value profiles when the
$i$th bidder is chosen to be the highest (and the other the second
highest).

In order to pin down $H(v)$, we consider $H(v)$ which makes the
seller indifferent across all (possible) reserve prices in $[0,1)$.
Suppose the seller charges a reserve price $p$ facing bidders whose
valuations are distributed as conjectured above. Then, he earns the
revenue:

\begin{align*}
\eta(p;F) & =\int_{V}\max\{v^{(2)},p\}\cdot1_{\{v^{(1)}>p\}}dF\\
 & =\begin{cases}
2\cdot\frac{1}{2}p(1-H(p)) & \text{if }p>\alpha\\
\alpha & \text{if }p\leq\alpha
\end{cases}\\
 & =\begin{cases}
p(1-H(p)) & \text{if }p>\alpha\\
\alpha & \text{if }p\leq\alpha.
\end{cases}
\end{align*}

For the seller to be indifferent across reserve prices, we require
that for all $p\in(\alpha,1)$:

\[
p(1-H(p))=\alpha,\text{ or }H(p)=\frac{p-\alpha}{p}.
\]

Since $\lim_{p\to1}H(p)=1-\alpha$, there must be a mass of $\alpha$
at $1$. Intuitively, for the seller to earn $\alpha$ even for $p$
arbitrarily close to 1, nature must put a mass of $\alpha$ at $1$.
Our conjecture for $H$ is therefore:

\[
H(v)=\begin{cases}
0 & \text{if }v\leq\alpha\\
\frac{v-\alpha}{v} & \text{if }v\in(\alpha,1)\\
1 & \text{if }v=1.
\end{cases}
\]
This distribution has density $h(v)=\frac{\alpha}{v^{2}}$ on $v\in(\alpha,1)$.

We can pin down $\alpha$ using the mean constraint on $F$. Each
bidder $i$ is chosen to be the highest bidder with probability $1/2$
, in which case his valuation is determined by $H$. With the remaining
probability, he is chosen to be the second highest bidder, with valuation
$\alpha$. Hence, the mean of each bidder's marginal distribution
is:

\[
m=\frac{1}{2}\alpha+\frac{1}{2}\left(\alpha+\int_{\alpha}^{1}vh(v)dv\right)=\alpha+\frac{1}{2}\int_{\alpha}^{1}\frac{\alpha}{v}dv=\alpha\left(1-\frac{1}{2}\ln(\alpha)\right),
\]
This equation pins down a unique solution for $\alpha$, as we will
later argue in the general $n$ bidder setting. When $m=1/2$, $\alpha=0.317$,
as is depicted in Figure \ref{fig:hyperplane_support}.

Next, we construct $G^{*}$ so that the conjectured $F^{*}$ will
be a best response by nature. If we take a step back, we can see that
for any given strategy $G$ by the seller, nature faces the following
problem:

\begin{align}
\inf_{F} & \int_{V}\phi(\mathbf{v};G)dF\label{eq:nature_program-1}\\
\text{s.t. } & \int_{V}v_{i}(\mathbf{v})dF=m,\quad i=1,2\nonumber \\
 & \int_{V}dF=1\nonumber \\
 & F\geq0,\nonumber 
\end{align}
where $v_{i}:V\to\mathbb{R}$ is the projection map on the $i$th
bidder's valuation and $\phi$ is the expected revenue from nature's
point of view, $\phi(\mathbf{v};G):=\int_{0}^{1}\max\{v^{(2)},p\}\cdot1_{\{v^{(1)}>p\}}dG.$

It is useful to observe that this is a linear program in $F$. Given
our conjecture for $F^{*}$, we can use the dual of this program to
verify the optimality of $F^{*}$ with respect to the expected revenue
function $\phi$ induced by the seller's strategy $G$. We use a complementary
slackness condition expressed in the following lemma. Note also that
this lemma applies generally to the $n$-bidder case, and in fact,
we will reuse this lemma to prove the optimality of our candidate
$F^{*}$ in the general case. 
\begin{lem}
\label{lem:duality}Given the seller's strategy $G$, $F^{*}$ is
an optimal solution for nature's problem if and only if $F^{*}$ is
feasible and there exists an affine function $L:V\to\mathbb{R}$ such
that

\begin{align*}
L(\mathbf{v}) & \leq\phi(\mathbf{v};G),\text{ }\forall\mathbf{v}\in V\text{ and }\\
\text{\emph{supp}}(F^{*}) & \subseteq\{\mathbf{v}:L(\mathbf{v})=\phi(\mathbf{v};G)\}.
\end{align*}
\end{lem}
\begin{proof}
We prove the ``if'' part. The proof of the ``only if'' part, which
requires the dual program, is relegated to the Appendix.

Suppose there exist feasible $F^{*}$ and an affine function $L$
satisfying the requirements. For any $F\in\mathcal{P}(V,m)$, this
implies

\[
\int_{V}\phi(\mathbf{v};G)dF^{*}=\int_{V}L(\mathbf{v})dF^{*}=\int_{V}L(\mathbf{v})dF\leq\int_{V}\phi(\mathbf{v};G)dF.
\]
The first equality follows because $\phi$ and $L$ coincide on the
support of $F^{*}$. The second equality follows because $L$ is affine
and $F$ and $F^{*}$ have the same mean. The final inequality follows
since $L(\mathbf{v})\leq\phi(\mathbf{v};G)$ for all $\mathbf{v}$.
Hence $F^{*}$ is an optimal solution for nature. 
\end{proof}
This lemma states that in order for our conjectured $F^{*}$ to be
a best response to the seller's strategy $G$, there must exist an
affine function $L(\mathbf{v})$ such that expected revenue function
$\phi$ induced by $G$ coincides with $L$ on the support of $F^{*}$,
and is everywhere above $L$. The key implication of this lemma is
that if we can find a strategy $G^{*}$ and an affine function $L$
such that the conditions of the lemma are satisfied, then $(G^{*},F^{*})$
is an equilibrium.\footnote{This duality lemma borrows from the approach used in \citet{scarf1957minmax}
in the context of inventory management, in which he constructs a supporting
quadratic polynomial to solve for the worst-case demand distribution
under a mean and variance constraint. This technique has also been
used in the context of monopoly pricing in \citet{carrasco2018optimalselling},
in an optimal transport program for an auction setting in \citet{He2019RobustlyOptimal},
and in Bayesian persuasion in \citet{Dworczak2019}. One difference
relative to Scarf and several other recent work is that the primal
objective function $\phi$ is itself endogenous. Unlike finding the
dual for a fixed $\phi$, we are simultaneously choosing $G$ and
$L$ to satisfy complementary slackness.} We use these conditions as a tool for verifying an equilibrium. In
particular, we will find $G$ that satisfies the conditions, given
the $F^{*}$ we already conjectured.

Given the symmetry of $F^{*}$, it is without loss to focus on $S_{1}$
for evaluating $\phi({\bf v};G)$, where $S_{1}$ is the subset of
the ``L''-shape support in which nature picks the bidder 1 to be
the highest bidder. On $S_{1}$, the expected revenue function will
be:

\begin{align*}
\phi(\mathbf{v};G) & =\int_{0}^{1}\max\{v^{(2)},p\}\cdot1_{\{v^{(1)}>p\}}dG\\
 & =\int_{0}^{v_{2}}v_{2}g(p)dp+\int_{v_{2}}^{v_{1}}pg(p)dp,
\end{align*}
assuming that $G$ admits density $g(p)$, which will be the case
as we show below. For the conjectured $F^{*}$ to be optimal, Lemma
\ref{lem:duality} requires that $\phi(\mathbf{v};G)$ must be equal
to some affine function $L(\mathbf{v})=c_{1}v_{1}+c_{2}v_{2}+d$ on
this segment. This implies that for all $(v,\alpha)\in S_{1}$:

\[
\frac{\partial\phi}{\partial v_{1}}(v,\alpha)=vg(v)=c_{1},\text{ or }g(v)=\frac{c_{1}}{v}.
\]

This logic applies symmetrically if nature had picked bidder 2 to
be the highest bidder instead, so we can consider $c_{1}=c_{2}=c$.
This pins down the form of the sender's strategy for $p\in[\alpha,1]$.
In order to satisfy the remaining conditions of the Lemma, we now
can find $c$ so that $\phi\geq L$ for all value profiles.

Next, consider any ${\bf v}=(v,\alpha)$, for $v>\alpha$. Since $(v,\alpha+\epsilon)$
is outside $\text{supp}(F^{*})$, for any $\epsilon>0$ or $\epsilon<0$,
Lemma 1 requires that $\phi(v,\alpha+\epsilon)\ge L(v,\alpha+\epsilon)$,
for any $\epsilon>0$ or $\epsilon<0$ and that $\phi(v,\alpha)=L(v,\alpha)$.
This implies that:

\[
\frac{\partial\phi}{\partial v_{2}}(v,\alpha)=\frac{\partial L}{\partial v_{2}}(v,\alpha)\implies G(\alpha)=c.
\]

%We must have $c=G(\alpha)$ or else $\phi(v,\alpha-\epsilon)<L(v,\alpha-\epsilon)$
%for $\epsilon$ small enough. 

Since $G$ must integrate to $1$, we find:

\[
1=\int_{0}^{1}dG(p)=G(\alpha)+\int_{\alpha}^{1}\frac{G(\alpha)}{p}dp=G(\alpha)(1-\ln(\alpha)),
\]
so:

\[
c=G(\alpha)=\frac{1}{1-\ln(\alpha)}.
\]

We will choose $G(0)=G(\alpha)=c$, so that the seller places the
entire mass of $c$ at $p=0$. This will be seen to be optimal for
the seller, although it is possible that there could be other payoff-equivalent
equilibria in which the seller may put density on $[0,\alpha]$. Note
that unlike the equilibrium for single buyer case in \citet{carrasco2018optimalselling},
the support of the seller and nature do not fully coincide. In other
words, with positive probability, the seller chooses a reserve price
that guarantees sale. As we will see later, this probability increases
as the number of bidders increases. The intuition is that as competition
increases, the seller can rely less on a ``binding'' reserve price
to guarantee revenue as competition between the bidders limits the
ability of nature to suppress revenue.

Finally, the condition $\phi=L$ at $(\alpha,\alpha)$ pins down the
constant $d=\frac{\alpha}{1-\ln(\alpha)}$ of $L({\bf v})$. Thus,
the strategy $G^{*}$ and the affine function $L$ are fully characterized
as follows:

\begin{align*}
G^{*}(p) & =\begin{cases}
\int_{\alpha}^{p}\frac{1}{r\left(1-\ln(\alpha)\right)}dr+\frac{1}{1-\ln(\alpha)} & \text{if }p>\alpha\\
\frac{1}{1-\ln(\alpha)} & \text{if }p\leq\alpha
\end{cases}\\
L(\mathbf{v}) & =\frac{1}{1-\ln(\alpha)}(v_{1}+v_{2})-\frac{\alpha}{1-\ln(\alpha)}.
\end{align*}

One can check that $\phi({\bf v};G^{*})\geq L(\mathbf{v})$ for all
$\mathbf{v}\in V$ and $\phi=L$ on the support of $F^{*}$, satisfying
the conditions of Lemma \ref{lem:duality}. The left panel of Figure
\ref{fig:hyperplane_support} visualizes how $\phi$ and $L$ satisfy
this requirement for the case of $m=1/2$. Lemma \ref{lem:duality}
guarantees that $F^{*}$ will be a best response to $G^{*}$. Recall
that $F^{*}$ was constructed to make the seller indifferent across
all reserve prices $p\in[\alpha,1)$ and thus $G^{*}$ is a best response
to $F^{*}$. Thus, we have found an equilibrium for the two-bidder
game with symmetric mean constraints.

Two differences emerge with the introduction of the second bidder
(compared to the monopoly pricing equilibrium in \citet{carrasco2018optimalselling}).
First, as mentioned earlier, the seller puts mass strictly below the
support of the marginal distribution of $F$. Second, the marginal
distribution of $F$ has two mass points, one at the upper bound $v=1$
(as in the single buyer equilibrium), and one at the lower bound $\alpha$
(unlike the single buyer equilibrium).\footnote{As mentioned before, the marginal distribution of the optimal $F^{*}$
is outside the class of marginal distributions considered in \citet{He2019RobustlyOptimal},
which does not allow for mass points. The authors also require that
$xf(x)$ is weakly increasing in $x$, where $f$ is the exogenously
given density. Even in the region where the marginal distribution
admits density $h$, the density fails this condition: $vh(v)=\frac{\alpha}{v^{2}}v=\frac{\alpha}{v}$,
which is decreasing in $v$.}

\subsection{Optimal reserve price with $n$ (possibly asymmetric) bidders}

In this section, we present a profile of candidate strategies $(G^{*},F^{*})$
that form a second-price auction equilibrium in the $n$-bidder case
with (possibly) asymmetric mean constraints. The equilibrium is a
natural extension of the two bidder, symmetric mean equilibrium. In
Appendix A.3, we will prove that these strategies satisfy the saddle
point inequalities.

\paragraph{Nature's strategy.}

Nature's worst case distribution retains the ``L''-shape structure,
randomly choosing a bidder to be the highest bidder (with value distributed
according to $H(v)$) while setting the valuations of the remaining
bidders to a lower bound $\alpha$. However, a few modifications are
introduced when there are bidders with asymmetric means. First, the
``L''-shape joint distribution may not include all $n$ bidders.
If there are bidders with means strictly less than $\alpha$,
they cannot be included in this distribution without violating these
constraints. Instead, nature will set the marginal distributions of
these bidders to a point mass on their means. Second, among the ``active''
bidders who are included in the distribution, nature selects one of
them to be the highest-value bidder with non-uniform probabilities.
We will see that higher mean bidders are more likely to be selected
as the highest-value bidder. 

We will proceed to describe the set of active bidders and the joint
distribution among the active and inactive bidders.

\paragraph{Selection of Active Bidders.}

Recall that the bidders are indexed such that $m_{1}\ge m_{2}\geq...\geq m_{n}$.
Suppose in the chosen distribution, there exists a cutoff $k$
such that bidders $i=1,...,k$, with mean $m_{i}\ge m_{k}$
are included in the $L$-shape distribution (the ``active'' bidders)
and the remaining bidders $k+1,...,n$ are excluded (the ``inactive''
bidders). Following the reasoning provided in Section 3.2, we postulate
that the lower bound of the support is $\alpha_{k}(\bar{m}_{k})$,
where $\bar{m}_{k}=\frac{1}{k}\sum_{i=1}^{k}m_{i}$ is the
average of the highest $k$ mean constraints (i.e. the mean of
the mean constraints among the active bidders) and $\alpha_{k}(x)$
is the solution of $x=\alpha\left(1-\frac{1}{k}\ln(\alpha)\right)$.
For this to be a consistent with the hypothesis, it must be the case
that $m_{i}>\alpha_{k}(\bar{m}_{k})$ for all $i\leq k$ and
$m_{r}\le\alpha_{k}(\bar{m}_{k})$ for all $r>k$. This reasoning
suggests the cutoff $k$ is determined by 

\begin{align}
k\equiv\min & \left\{ \ell:m_{i}>\alpha_{\ell}(\bar{m}_{\ell})\text{ for }i\leq\ell,m_{r}\leq\alpha_{\ell}(\bar{m}_{\ell})\text{ for }r>\ell\right\} ,\label{eq:asymmetric_cutoff}
\end{align}

For example, if there are $n=3$ bidders with $m_{1}=0.6$, $m_{2}=0.5$,
and $m_{3}=0.1$, the cutoff $k$ will be $2$ since $\bar{m}_{2}=\frac{m_{1}+m_{2}}{2}=0.55$,
$\alpha_{2}(\bar{m}_{2})\approx0.366$, and $m_{1},m_{2}>\alpha_{2}(\bar{m}_{2})$
with $m_{3}<\alpha_{2}(\bar{m}_{2})$. There will be two active bidders,
since the mean of the third bidder's valuation is too low. Note that
if the means of all $n$ bidders are symmetric $m_{1}=...=m_{n}$,
then the cutoff is $k=n$ and all bidders are active. 

We find that given any $n$ bidders with means $\{m_{i}\}_{i=1}^{n},$
we can always find such a cutoff through a simple iterative procedure.
\begin{lem}
\label{lem:asymmetric_cutoff}The cutoff $k$ in (\ref{eq:asymmetric_cutoff})
always exists. 
\end{lem}
\begin{proof}
See Appendix \ref{sec:Appendix:proof_asymmetric_cutoff}. 
\end{proof}
Note in addition that $\alpha_{k}(x)$ is well-defined for any $k\in[n]$
and $x\in[0,1]$. To see this, note that when $\alpha=0$, $\alpha\left(1-\frac{1}{k}\ln(\alpha)\right)=0$
and when $\alpha=1$, this is equal to $1$. Since the expression
is continuous and strictly increasing, there exists a unique $\alpha$
such that this equation holds. Explicitly, $\alpha_{k}(\bar{m}_{k})$
is:

\[
\alpha_{k}(\bar{m}_{k})=\exp\left(k+W_{-1}\left(-\frac{k\bar{m}_{k}}{e^{k}}\right)\right),
\]
where $W_{-1}$ is the lower branch of the Lambert $W$ function.
When all $n$ bidders have means equal to $m$, the equation for $\alpha$
becomes

\begin{equation}
m=\alpha\left(1-\frac{1}{n}\ln(\alpha)\right).\label{eq:alpha_symmetric_means}
\end{equation}

\begin{figure}[h]
\caption{\label{fig:cutoff_k}Support of $F^{*}$ with symmetric and asymmetric
mean constraints}

\bigskip{}

\begin{centering}
\includegraphics[width=4.8in]{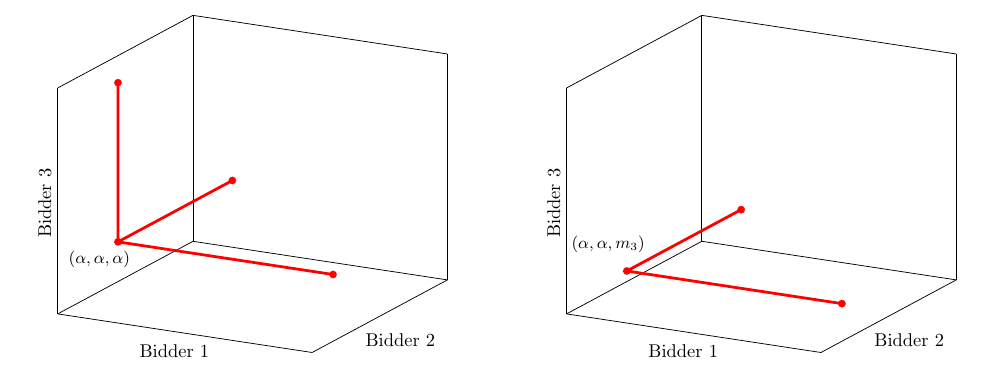}
\par\end{centering}
\bigskip{}

{\scriptsize{}{}{{\scriptsize{}{}Notes:}} On the left panel
is an example of $F^{*}$ with $n=3$ bidders with symmetric mean
constraints. On the right panel is an example of $F^{*}$ with $n=3$
bidders but only $k=2$ active bidders, in which $m_{3}<\alpha_{2}(\bar{m}_{2})$.}{\scriptsize\par}
\end{figure}

\paragraph{Worst-case Distribution.}

With the cutoff $k$ from (\ref{eq:asymmetric_cutoff}) in hand, we
can now construct nature's worst-case distribution $F^{*}$. In $F^{*}$,
each inactive bidder $j\in\{k+1,...,n\}$ draws value $m_{j}$ with
probability one. Each active bidder $i\in\{1,...,k\}$ is selected
to be the highest bidder with probability $\theta_{i}$, in which
case all other active bidders draw valuations $\alpha$ and the highest
bidder's value is distributed according to a CDF $H(v)$. A typical
value profile will be as follows:

\begin{equation}
{\bf v}_{i}(v)\equiv(\underbrace{\alpha,\alpha,...,v,...,\alpha}_{k\text{ active bidders}},\underbrace{m_{k+1},m_{k+2},...,m_{n}}_{n-k\text{ inactive bidders}}),\label{eq:value-profile}
\end{equation}
where bidder $i$ has the highest value of $v\sim H(v)$.

As in the two-bidder case, the highest bidder's valuation is distributed
according to the CDF

\[
H(v)=\begin{cases}
0 & \text{if }v\leq\alpha\\
\frac{v-\alpha}{v} & \text{if }v\in(\alpha,1),\\
1 & \text{if }v=1.
\end{cases}
\]
As with the two bidder case, $H(v)$ ensures that the seller is indifferent
with his choice of reserve price within the support of $H$. 

The probabilities $\theta_{i}$ that an active bidder $i$ is selected
as the highest-value bidder are chosen to satisfy the mean constraints\footnote{The mean constraints for the $n-k$ bidders are satisfied automatically,
since their valuations are always equal to their means.} 

\begin{equation}
(1-\theta_{i})\alpha+\theta_{i}\left(\alpha+\int_{\alpha}^{1}vh(v)dv\right)=m_{i}.\label{eq:bidder_i_mean}
\end{equation}

Solving for $\theta_{i}$, we obtain

\[
\theta_{i}=\frac{m_{i}-\alpha}{\alpha\ln(\frac{1}{\alpha})}.
\]

Naturally, the higher the mean $m_{i}$, the higher the selection
probability $\theta_{i}$. In other words, $\theta_{1}\geq\theta_{2}\geq...\geq\theta_{k}$.
Nature adjusts the selection probabilities to ``soak up'' all the
asymmetries, while ensuring that the value distribution of the highest
bidder is the same (given by $H$) regardless of which bidder is chosen.
In an important sense, Nature eliminates asymmetries from the perspective
of the seller, which ultimately eliminates the need for the seller
to treat bidders differently.

Finally, summing across equations (\ref{eq:bidder_i_mean}) for across
all $k$ active bidders, we obtain the following condition for the
lower bound $\alpha$,

\begin{align*}
\frac{1}{k}\sum_{i=1}^{k}m_{i}=\bar{m}_{k} & =\alpha\left(1-\frac{1}{k}\ln(\alpha)\right),
\end{align*}
which confirms that the lower bound $\alpha$ is indeed equal to $\alpha_{k}(\bar{m}_{k})$.

Altogether, nature's strategy $F^{*}$ is:

\begin{align}
F^{*}(U)\equiv & \sum_{i=1}^{k}\theta_{i}\cdot H(\text{Proj}_{i}(U))\label{eq:nature_strategy}
\end{align}
for any measurable subset $U$ of $\text{supp}(F^{*})=\bigcup_{i\in[k]}\{{\bf v}_{i}(v):v\in(\alpha,1]\}$,
where ${\bf v}_{i}(v)$ are the value profiles defined in (\ref{eq:value-profile}),
$\text{Proj}_{i}:U\to\mathbb{R}$ is the projection map on to the
$i$th bidder's valuation, and (overloading notation) $H$ is the
measure induced by the CDF $H$.

\paragraph{Seller's Strategy.}

The seller's strategy is to randomize his reserve price according
to the CDF:

\begin{equation}
G^{*}(p)=\begin{cases}
\int_{\alpha}^{p}\frac{1}{r\left(k-1-\ln(\alpha)\right)}dr+\frac{k-1}{k-1-\ln(\alpha)} & \text{if }p>\alpha\\
\frac{k-1}{k-1-\ln(\alpha)} & \text{if }p\leq\alpha.
\end{cases}\label{eq:seller_strategy}
\end{equation}
In other words, the seller chooses $p\in[\alpha,1]$ with density
$g(p)=\frac{1}{p\left(k-1-\ln(\alpha)\right)}$ and places a point
mass of $\frac{k-1}{k-1-\ln(\alpha)}$ on the reserve price $p=0$.
As in the two-bidder case, the density is constructed to prevent nature
from deviating from $F^{*}$ by satisfying the conditions of Lemma
\ref{lem:duality}.

Note that from the seller's perspective, the particular asymmetries
between bidders do not matter; the optimal reserve price only depends
on the number of active bidders $k$ and the lower bound $\alpha$.
This is because all the payoff-relevant features remain the same regardless
of which bidder is picked as the winner: the highest bidder's value
is always drawn from $H(v)$ and the second-highest value is always
$\alpha$. Hence, the seller's optimal reserve price is symmetric
across bidders.

The pair of strategies $(G^{*},F^{*})$ satisfy the saddle point condition,
proving that $S_{G^{*}}$ is revenue guarantee optimal among all second-price
auction mechanisms.
\begin{thm}
\label{thm:saddle_point}The strategy profile $(G^{*},F^{*})$ satisfies
the saddle point condition in (\ref{eq:saddle_point}) and therefore
are an equilibrium of the zero sum game. The revenue guarantee is
$\Psi(G^{*},F^{*})=\alpha_{k}(\bar{m}_{k})$. 
\end{thm}
\begin{proof}
See Appendix \ref{sec:Proof}. 
\end{proof}
We explore several implications of this theorem. First, suppose there
are $n$ bidders with symmetric means $\{m_{i}\}_{i=1}^{n}=m$.\footnote{We consider the symmetric mean case to analyze asymptotics, since
this allows us to take the number of bidders $n\to\infty$, without
worrying about the number of active bidders, which may remain finite
as $n\to\infty$.} Let $G_{n}^{*}$ and $F_{n}^{*}$ be the equilibrium strategy in
this auction. 
\begin{cor}
\label{cor:limit}Given a sequence of auctions with symmetric mean
bidders $\{m_{i}\}_{i=1}^{n}=m$, there exists a sequence of equilibria
$\left(G_{n}^{*},F_{n}^{*}\right)$ such that the random reserve price
$p_{n}\sim G_{n}^{*}$ converges to zero in probability as $n\to\infty$.
Furthermore, the optimal revenue guarantee $\Psi(G_{n}^{*},F_{n}^{*})=\alpha_{n}$
increases in $n$ and converges to $m$ at rate $O(1/n)$. The rent
obtained by the highest value bidder is $\alpha_n \log (1/\alpha_n)$, which
converges to $m \log(1/m)\leq 1/e$.
\end{cor}
\begin{proof}
We define $\alpha_{n}$ to be the solution of $m=\alpha\left(1-\frac{1}{n}\ln(\alpha)\right)$
when there are $n$ bidders and symmetric means equal to $m$. First,
we observe that $\alpha_{n}$ is increasing in $n$. This is because
the RHS of (\ref{eq:alpha_symmetric_means}) is decreasing in $n$,
for any $\alpha\in(0,1)$. Hence, as $n$ increases, the RHS shifts
down. Since RHS is an increasing function of $\alpha$, the $\alpha$
satisfying (\ref{eq:alpha_symmetric_means}) must increase as $n$
increases. Furthermore, as $n$ goes to infinity the RHS converges
to $\alpha$, which implies that $\alpha_{n}$ converges to $m$.

The convergence of the reserve price results from observing that the reserve
price $p_{n}\sim G_{n}^{*}$ with $G_{n}^{*}$ defined in (\ref{eq:seller_strategy})
converges in probability to zero as $n\to\infty$, as for any $\epsilon>0$

\[
\mathbb{P}(p_{n}>\epsilon)=1-G_{n}^{*}\left(\epsilon\right)\leq1-G_{n}^{*}\left(\alpha_{n}\right)=\frac{-\ln(\alpha_{n})}{n-1-\ln(\alpha_{n})}\rightarrow0.
\]

We can see that that since $\alpha_{n}$ solves $m=\alpha_{n}\left(1-\frac{1}{n}\ln(\alpha_{n})\right)$
we have that $m-\alpha_{n}=\frac{\alpha_{n}\ln\left(\frac{1}{\alpha_{n}}\right)}{n}\leq\frac{1}{ne}$,
so the rate of convergence is $O(\frac{1}{n})$.

Finally, the mean of the distribution of the highest-value bidder $H(v)$ can be computed
as $\alpha(1 + \log(1/\alpha))$. The highest bidder pays $\alpha$ and is thus
left with $\alpha \log(1/\alpha)$.

\end{proof}
\begin{cor}
\label{cor:fosd_highest_bidder} Fix two situations with $n$ and
$n'>n$ bidder, where all bidders have the same mean $m$. The distribution
of the highest bidder's value $H_{n'}$ first-order stochastically
dominates $H_{n}$. 
\end{cor}
\begin{proof}
The difference $H_{n}(v)-H_{n'}(v)$ for all $v\in[0,1]$ is as follows:

\[
H_{n}(v)-H_{n'}(v)=\begin{cases}
0 & \text{if }v\leq\alpha_{n}\\
\frac{v-\alpha_{n}}{v} & \text{if }v\in(\alpha_{n},\alpha_{n'}]\\
\frac{\alpha_{n'}-\alpha_{n}}{v} & \text{if }v\in(\alpha_{n'},1)\\
0 & \text{if }v=1.
\end{cases}
\]

As shown in the proof of Corollary \ref{cor:limit}, $\alpha_{n}$
is increasing in $n$. Hence, for all $v\in[0,1]$, $H_{n}(v)\geq H_{n'}(v)$. 
\end{proof}
Corollary \ref{cor:limit} says that as the number of bidders grows
large, the equilibrium revenue converges to the mean of a single bidder's
valuation, which implies that the seller extracts a \emph{single}
bidder's surplus in the limit. Additionally, it states that there
exists a sequence of equilibria such that the seller's reserve price
converges to zero in probability, as the number of bidders goes to
infinity (Figure \ref{fig:strategies} shows the effect of increased
competition from $n=2$ to $n=10$). The asymptotic behavior of the
equilibrium reserve price strategy is consistent with the widespread
observation that reserve prices tend to be lower than is prescribed
by the standard (Bayesian) theory under the estimated distribution
of value profiles, as noted by \citet{tamer2003} in timber auctions,
\citet{bajari2003winners} in eBay auctions, and \citet{mcaffee2002}
in real estate auctions.

\begin{figure}[h]
\caption{\label{fig:strategies} Strategies for nature and the seller when
$n=2$ and $n=10$ ($m=\frac{1}{2}$).}

\bigskip{}

\begin{centering}
\hspace{0.1cm}\includegraphics[width=4.8in]{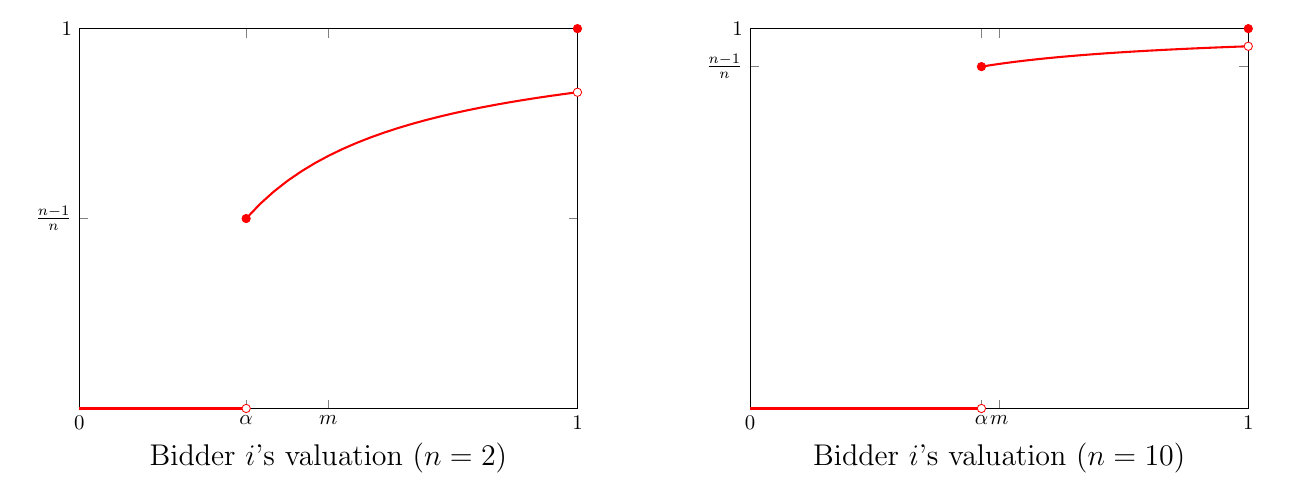} 
\par\end{centering}
\bigskip{}

\begin{centering}
\includegraphics[width=4.84in]{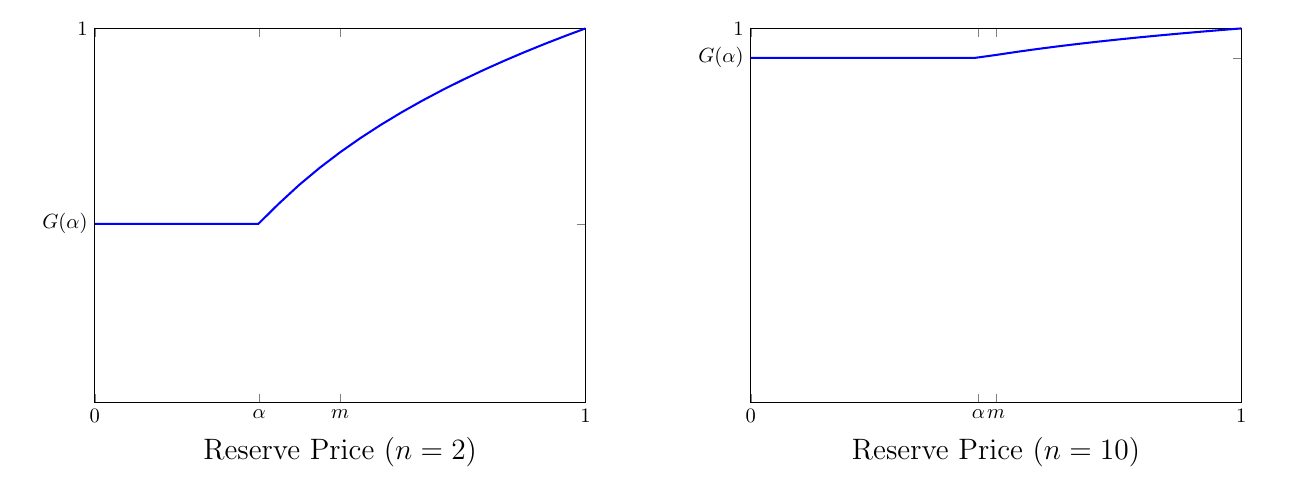} 
\par\end{centering}
\bigskip{}

{\scriptsize{}{}{{\scriptsize{}{}Notes:}} These are the equilibrium
strategies for nature and the seller when $n=2$ and $n=10$ and all
bidders have symmetric means $m=1/2$. In red is the CDF of the marginal
distribution of nature's equilibrium strategy $F^{*}$. In blue is
$G^{*}$, or the CDF of the seller's strategy over reserve price.
When $n=2$, $\alpha\approx0.317$ and when $n=10$, $\alpha\approx0.464$.
Observe that as the number of bidders increases, the seller's mass
on $p=0$ increases towards 1, and $\alpha$ increases towards the
mean $m$.}{\scriptsize\par}

\end{figure}

A key feature of this equilibrium is that the seller randomizes in
reserve price. Since nature can only pick one worst-case distribution,
randomization allows the seller to hedge against nature. Randomization
protects the seller against states of the world that are adverse for
a particular reserve price, and it forces nature to minimize revenue
across multiple value realizations.

As the number of bidders increases however, the role of reserve price
disappears, as the probability of binding reserve prices vanishes.
In essence, competition between bidders appears substitutes for a
reserve price. The intuition of this result is that competition constrains
the seller's worst case belief: as the number of bidders increases,
nature is compelled to increase the value of the second-highest bidder
$\alpha$ to maintain the mean constraint. Furthermore, the marginal
distribution of each bidder converges to a point mass on $\alpha$
(as each bidder is the second-highest bidder with probability $\frac{n-1}{n}$).
In the limit, the worst-case joint distribution becomes the point
mass distribution on the value profile consisting of $\alpha$ for
all bidders, which itself converges to $m$. The optimal strategy
for the seller facing a point mass on ${\bf m}=(m,...,m)$ is to charge
a reserve price of zero.

\section{\label{sec:Optimal-Competitive-Mechanism}Optimal Competitive Mechanism}

In this section, we explore the performance of the optimal second-price
auction within a broad class of mechanisms and show that the second-price
auction with an optimal reserve price is revenue-guarantee optimal
within this class.

We begin with some definitions. Given a fixed value distribution $F$,
we say that a realized bidder-type $(i,v_{i})$ is \emph{clearly low}
if (1) with probability $1$ there exists some bidder $j\neq i$ with a
higher valuation than $v_i$ and (2) conditional on bidder $j$ having a higher
valuation, bidder $i$'s distribution is a point mass on $v_{i}$.
Formally, we define the set of clearly low bidder-values to be:

\[
\xi(F)=\{(i,v):\exists j\neq i\text{ s.t. }\mathbb{P}_{F}(v_{j}>v|v_{i}=v)=1\text{ and }\mathbb{P}_{F}(v_{i}=v|v_{j}>v)=1\}.
\]

The event has a clear implication for what bidders know when they
have a common knowledge about the value distribution $F$. In that
case, whenever a bidder realizes a clearly low value $v$, it becomes
common knowledge between $i$ and some bidder $j$ (whose identity
may be unknown to $i$) that bidder $j$ has a strictly higher value
than $v$ and bidder $i$ has exactly $v$.

We say that a bidder-type \emph{never wins at}\textbf{ $\Gamma(M,I,F)$}
if he is assigned the good with probability zero in every Bayes-Nash
equilibrium in undominated strategies of the game induced by $(M,I,F)$.
The class of \emph{competitive} mechanisms is defined as:

\[
\widehat{\mathcal{M}}=\{M\in\mathcal{M}:\forall F\in\mathcal{P}(V),(i,v)\in\xi(F)\text{ never wins at }\Gamma(M,I_{0},F)\}.
\]

This definition imposes a very mild restriction on the set of possible
mechanisms. First, the definition rules out behavior of the mechanism
only under the UCP information structure $I_{0}$.
Second, it only restricts allocation in rather atypical realizations
of the value distribution when two bidders both know that one has
a higher value than the other with complete certainty (and the higher-value
bidder must know the lower-value bidder's value exactly). These realizations
do not exist in many typical value distributions, e.g. if values are
drawn iid from a non-degenerate distribution, and no restriction is imposed for
such distributions.

Given the mildness of this restriction, the class of competitive
mechanisms encompasses many standard auction mechanisms without reserve
prices or with a common reserve price. This includes all efficient
mechanisms (which allocate the good to the highest value bidder) but
also includes some inefficient mechanisms, such as the first-price
auction.
\begin{lem}
\label{fpa-competitive}The second-price auction and the first-price
auction, both possibly with a common (possibly random) reserve price,
are competitive mechanisms.
\end{lem}
\begin{proof}
See Appendix \ref{sec:Appendix:proof_fpa}.
\end{proof}
Evidently, the class of competitive mechanisms is not the largest
class of mechanisms for which the second-price auction is revenue-guarantee
optimal. For example, lottery mechanisms (which give away the good
to a random bidder) or an all-pay auction with reserve prices uniformly
above or below $\alpha$ are not competitive mechanisms but are still
revenue dominated by the optimal second-price auction given $F^{*}$.
Rather, we focus on this particular class of mechanisms since it prescribes
a natural and weak restriction on admissible mechanisms.

Our main result is that the optimal second-price auction found in
section \ref{sec:second-price-auction} obtains the optimal revenue
guarantee within the class of competitive mechanisms. 
\begin{thm}
\label{thm:optimal_mechanism}The second-price auction mechanism $S_{G^{*}}$
with the optimal reserve price $G^{*}$ in (\ref{eq:seller_strategy})
obtains the optimal revenue guarantee out of all competitive mechanisms.
In other words, for any $M\in\widehat{\mathcal{M}}$ and $(F,I)\in\mathcal{P}(V,\{m_{i}\}_{i=1}^{n})\times\mathcal{I}$.

\[
\mathcal{R}(M,I_{0},F^{*})\leq\mathcal{R}(S_{G^{*}},I_{0},F^{*})\leq\mathcal{R}(S_{G^{*}},I,F),
\]
where $F^{*}$ is the worst-case distribution in (\ref{eq:nature_strategy})
and $I_{0}$ is the UCP information structure.
These inequalities imply:

\[
S_{G^{*}}\in\arg\max_{M\in\hat{M}}\inf_{F,I}\mathcal{R}\left(F,I,M\right).
\]
\end{thm}

\begin{proof} See Appendix \ref{sec:Appendix:proof_theorem_2}.
\end{proof}

As mentioned before, although we restrict ourselves to the set of
competitive mechanisms, this by no means captures all of the mechanisms
within which $S_{G^{*}}$ revenue dominates, given $(F^{*},I_{0})$.
In fact, we can observe from the proof that $S_{G^{*}}$ revenue-guarantee
dominates any mechanism for which $\mathcal{R}(M,I_{0},F^{*})\leq\alpha$.

At the same time, it is not true that $S_{G}$ is optimal in the fully
unrestricted class of mechanisms. When facing $(F^{*},I_{0})$, the
seller can do strictly better with a mechanism that is not competitive. 
\begin{example}
Consider a setting with two bidders with identical means $m\in(0,1)$,
with valuations distributed by $F^{*}$. 

Consider the following mechanism. Both bidders are asked to report
who the clearly low bidder is (the bidder who draws $\alpha$) and
a bidder who reports himself not to be reports his value. If there
is no consensus on the clearly low bidder, then neither bidder gets
the good and there is no payment. If bidder $i$ is a consensus low
bidder, then bidder $j \neq i$ gets the good if and only if his value is
strictly above $1-\epsilon$ at price $1-\epsilon$. Otherwise, the
low bidder gets the good at price $\alpha-\epsilon$. For $\epsilon>0$
sufficiently small, it is a weak dominant strategy for the low bidder
to report truthfully. Further, it is a best response for the high
bidder to report truthfully about the other bidder being clearly low
and about true value. Indeed, the strategy is undominated. 

Hence, the expected revenue from the equilibrium is: 

\begin{align*}
\sum_{i=1}^{2}\mathbb{E}_{F^{*}}[t_{i}({\bf v})] & =\frac{1}{2}\sum_{i=1}^{2}\left[(1-H(1-\epsilon))(1-\epsilon)+H(1-\epsilon)(\alpha-\epsilon)\right]\\
 & \geq\alpha(1-\epsilon)+\frac{1-\epsilon-\alpha}{1-\epsilon}(\alpha-\epsilon).
\end{align*}
As $\epsilon\to0$, the RHS converges to $\alpha+(1-\alpha)\alpha>\alpha$
so for $\epsilon$ small enough the revenue exceeds $\alpha$.

We can make a few remarks about the mechanism. First, the above mechanism
is implausible as it involves extreme discrimination and is very sensitive
to the particular structure of the value distribution. It is never
robustly optimal itself. Nature could simply choose the two point distribution
placing mass at $\alpha-2\epsilon$ and some $z\in(m,1-\epsilon)$, which
would admit an equilibrium in which no sale arises. The second-price
auction does not suffer from such instabilities and does not discriminate
between bidders even if the seller has ex-ante knowledge of asymmetries
between the bidders.

Second, the above mechanism is also susceptible to collusion by bidders.
Facing the mechanism, the bidders would find it optimal to coordinate
their reports so that the high bidder is the low bidder and the low
bidder reports less than $1-\epsilon$ for his value, so that the
high value bidder always gets the good at price $\alpha-\epsilon$.
The bidder surplus would be sufficiently high so that such a side contract
can be mutually beneficial (and incentive compatible since they know
who the low bidder is as common knowledge). By contrast, the second
price auction does not leave any room for collusion, since the allocation
is efficient among the bidders, conditional on it being allocated
to somebody.
\end{example}

\section{\label{sec:Conclusion}Conclusion}

We have studied the seller's optimal selling mechanism to multiple
potential buyers, subject to moment conditions of the distribution
of bidders' values. We have identified the worst-case distribution
chosen by nature and the optimal reserve price strategy of the seller
as the equilibrium of a zero-sum, simultaneous move game. We have
found that among a wide class of Bayesian mechanisms, a second-price
auction with a random, symmetric reserve price obtains the optimal
revenue guarantee even in the presence of a priori asymmetric means
for the bidders. In addition, we have shown that when bidders have identical means
and as the number of bidders increases,
the optimal second-price auction involves a non-binding reserve price
and the revenue guarantee converges to the best possible revenue guarantee
at rate $O(1/n)$.

\bibliographystyle{aer}
\bibliography{bib_nash_robust}

\appendix
%dummy comment inserted by tex2lyx to ensure that this paragraph is not empty

\section{Proofs}

\subsection{Proof of ``Only if'' of Lemma \ref{lem:duality}}
\begin{proof}
To prove the converse of Lemma \ref{lem:duality}, we construct the
dual of nature's minimization problem in (\ref{eq:nature_program-1}):

\begin{align*}
\max_{\gamma_{i},\eta} & \sum_{i=1}^{n}\gamma_{i}\cdot m+\eta\\
\text{s.t. } & \sum_{i=1}^{n}\gamma_{i}v_{i}(\mathbf{v})+\eta\leq\phi(\mathbf{v};G),\text{ }\forall\mathbf{v}\in V.
\end{align*}
The dual variables $\gamma_{i}$ and $\eta$ together belong in $\mathbb{R}^{n+1}$.
This representation follows the dual of the semi-infinite linear program
in \citet{AndersonNash}.

The primal is bounded above by $1$, and the measure $F=\delta_{\mathbf{m}}$,
which assigns all mass to the mean point $\mathbf{m}=(m,m,...,m)$,
is in the interior of the primal cone. By Theorem 3.12 in \citet{AndersonNash},
strong duality holds.

The dual solution $(\gamma^{*},\eta^{*})$ defines an affine function
$L(\mathbf{v})=\left\langle \gamma^{*},\mathbf{v}\right\rangle +\eta^{*}$.
Given the primal solution, we can rewrite the dual in terms of this
affine function.

\begin{align*}
\max_{L} & \int_{V}L(\mathbf{v})dF\\
\text{s.t. } & L(\mathbf{v})\leq\phi(\mathbf{v};G),\text{ }\forall\mathbf{v}\in V.
\end{align*}

The constraint of the dual guarantees that the affine function is
below the expected revenue function. There is no duality gap, which
means:

\[
\int_{V}L(\mathbf{v})dF^{*}=\int_{V}\phi(\mathbf{v};G)dF^{*}.
\]

Therefore, points at which $L(\mathbf{v})<\phi(\mathbf{v};G)$ must
only occur on a measure zero set with respect to $F^{*}$. By the
definition of the support, we obtain:

\[
\forall\mathbf{v}\in\text{supp}(F^{*}),\text{ }L(\mathbf{v})=\phi(\mathbf{v};G).
\]
\end{proof}

\subsection{\label{sec:Appendix:proof_asymmetric_cutoff}Proof of Lemma \ref{lem:asymmetric_cutoff}}
\begin{proof}
We can use a straightforward algorithm to find $k$.

First, we begin with $i:=2$, the bidder with the second-highest mean
constraint. We calculate $\bar{m}_{i-1}$ and $\alpha_{i-1}(\bar{m}_{i-1})$
and check whether $m_{i}\leq\alpha_{i-1}(\bar{m}_{i-1})$. If so,
then we set $k=i-1$. If not, then we increment the index $i:=i+1$
and repeat the same procedure. If $i$ reaches $n+1$, then $k=n$.

There are finitely many bidders so this algorithm will terminate.
Since $k$ is returned by the algorithm only if $m_{k+1}\leq\alpha_{k}(\bar{m}_{k})$
(or if $k$ is the last bidder), for all bidders with index $p$ greater
than $k$, $m_{p}\leq m_{k+1}\leq\alpha_{k}(\bar{m}_{k})$.

It remains to show that $m_{j}>\alpha_{k}(\bar{m}_{k})$ for all $j\leq k$.
It must be true that $m_{k}>\alpha_{k-1}(\bar{m}_{k-1})$, otherwise
$k-1$ would have been the outputted cutoff. It suffices to show that
$m_{k}>\alpha_{k}(\bar{m}_{k})$. Suppose to the contrary that:

\[
\alpha_{k-1}(\bar{m}_{k-1})<m_{k}\leq\alpha_{k}(\bar{m}_{k}).
\]

The RHS of the equation $x=\alpha\left(1-\frac{1}{\ell}\ln(\alpha)\right)$
is increasing in $\alpha$, which means that for the equation satisfied
by $\alpha_{k-1}(\bar{m}_{k-1})$ (shortening notation to $\alpha_{k-1}$):

\begin{equation}
\bar{m}_{k-1}=\alpha_{k-1}\left(1-\frac{1}{k-1}\ln(\alpha_{k-1})\right)<m_{k}\left(1-\frac{1}{k-1}\ln(m_{k})\right),\label{eq:mk-1}
\end{equation}
and likewise

\begin{equation}
\bar{m}_{k}=\alpha_{k}\left(1-\frac{1}{k}\ln(\alpha_{k})\right)\geq m_{k}\left(1-\frac{1}{k}\ln(m_{k})\right).\label{eq:mk}
\end{equation}

We can multiply (\ref{eq:mk-1}) by $k-1$ and (\ref{eq:mk}) by $k$
and subtract the former from the latter to obtain:

\begin{align*}
k\bar{m}_{k}- & (k-1)\bar{m}_{k-1}>m_{k}\\
 & m_{k}>m_{k},
\end{align*}
which is a contradiction. Hence for all $j\leq k$, $m_{j}\geq m_{k}\geq\alpha_{k}(\bar{m}_{k})$.

The $k$ returned by the algorithm satisfies both conditions required
for the cutoff and furthermore, it will be the minimum index that
does so since we iterate from the smallest index $1$. 
\end{proof}

\subsection{\label{sec:Proof}Appendix: Proof of Theorem \ref{thm:saddle_point}}

To prove Theorem \ref{thm:saddle_point}, we will first show in Section
\ref{subsec:Seller's-problem} that the seller's strategy $G^{*}$
is optimal given $F^{*}$. In Section \ref{subsec:Nature's-problem},
we will prove that nature's strategy $F^{*}$ is optimal given $G^{*}$.

\subsubsection{\label{subsec:Seller's-problem}Seller's problem}

First, we show that for the $F^{*}$ described above, $G^{*}$ is
a best response by the seller. In other words, we will show the following
inequality holds for all $G\in\mathcal{P}\left([0,1]\right)$:

\[
\Psi(G,F^{*})\leq\Psi(G^{*},F^{*}).
\]

\begin{proof}
Recall that $F^{*}$ has the following form:

\[
F^{*}(U)=\sum_{i=1}^{k}\theta_{i}\mathbf{1}_{S_{i}}\cdot H(\textrm{Proj}_{i}(U)),
\]
where $k$ is the cutoff index from Lemma \ref{lem:asymmetric_cutoff},
which determines the number of active bidders. The distribution of
the highest bidder's value $H(v)$ is:

\[
H(v)=\begin{cases}
0 & \text{if }v\leq\alpha\\
\frac{v-\alpha}{v} & \text{if }v\in(\alpha,1)\\
1 & \text{if }v=1.
\end{cases}
\]

Using the same notation as in the two bidder case, we let $\eta(p;F)$
denote the seller's expected revenue under reserve price $p$ given
nature's strategy $F.$

\[
\eta(p;F):=\int_{V}\max\{v^{(2)},p\}\cdot1_{\{v^{(1)}>p\}}dF.
\]

Similar to the two bidder example, we will proceed to show that our
candidate $F^{*}$ makes the seller indifferent across all reserve
prices (at least in $[0,1)$). Under $F^{*}$, if the seller sets
the reserve price below $\alpha$, a sale will always occur and the
seller will receive the second highest value, $\alpha$. If the seller
sets the reserve price above $\alpha$, then he will make a sale with
probability $H(p)$ and receive the reserve price $p$ as revenue.
Hence, the seller's payoff under $F^{*}$ from a given reserve price
is:

\begin{align*}
\eta(p;F^{*}) & =\begin{cases}
\sum_{i=1}^{k}\theta_{i}\cdot p(1-H(p)) & \text{if }p>\alpha\\
\alpha & \text{if }p\leq\alpha
\end{cases}\\
 & =\begin{cases}
p(1-H(p)) & \text{if }p>\alpha\\
\alpha & \text{if }p\leq\alpha.
\end{cases}
\end{align*}

Note that the $n-k$ inactive bidders do not affect the seller's expected
revenue, as the second-highest value $\alpha$ is always greater than
their valuations. Note also that the number of bidders and their selection
probabilities $\theta_{i}$ only affect the seller through their impact
on $\alpha$. With the exception of $\alpha$, the expected revenue
to the seller $\eta$ for any reserve price $p$ is exactly the same
as in the two-bidder, symmetric mean example.

For any reserve price $p\in[0,1)$, the revenue to the seller will
be:

\[
p(1-H(p))=p\left(1-\frac{p-\alpha}{p}\right)=\alpha.
\]

When $p=1$, the revenue may be smaller, depending on the tie breaking
rule. Observe however that $G^{*}$ does not put any mass at 1, so

\[
\Psi(G^{*},F^{*})=\alpha\geq\Psi(G,F^{*}),
\]
for any $G\in\mathcal{P}\left([0,1]\right)$. Hence, $G^{*}$ in Theorem
\ref{thm:saddle_point} is a best response. 
\end{proof}

\subsubsection{\label{subsec:Nature's-problem}Nature's problem}

Now we will prove that for the $G^{*}$ in Theorem \ref{thm:saddle_point},
$F^{*}$ is a best response by nature. We will show that for any $F\in\mathcal{P}(V,m)$,

\[
\Psi(G^{*},F^{*})\leq\Psi(G^{*},F).
\]

\paragraph{Equivalence with a $k$-bidder auction.}

Since the $n-k$ inactive bidders do not affect the revenue, we can
prove the saddle point inequality using the $F^{*}$ distribution
containing only the $k$ active bidders, which simplifies the problem.
Consider a new auction $(V_{k},\{m_{i}\}_{i=1}^{k})$, in which we
include only the $k$ active bidders in the ``L''-shape joint distribution
of $F^{*}$.

For any distribution of value profiles in the original $n$-bidder
auction $F\in\mathcal{P}(V,\{m_{i}\}_{i=1}^{n})$, we can project
this distribution on the $k$ active bidders and obtain $F_{k}\in\mathcal{P}(V_{k},\{m_{i}\}_{i=1}^{k})$.
This distribution is exactly identical to $F$ for the $k$ active
bidders, but no longer includes the $n-k$ inactive bidders.

We see that from nature's point of view, removing these bidders is
without loss.
\begin{lem}
$\Psi(G^{*},F^{*})=\Psi(G^{*},F_{k}^{*})\text{ and }\Psi(G^{*},F_{k})\leq\Psi(G^{*},F).$
\end{lem}
\begin{proof}
For any general $F$ and $F_{k}$, $\Psi(G,F)\geq\Psi(G,F_{k})$ for
any reserve price $G$. All else being equal, including more bidders
can only increase the revenue for the seller. Moreover, under $F^{*}$
all $n-k$ inactive bidders draw valuations below $\alpha$ so the
revenue is identical whether the bidders are included or removed entirely.
\end{proof}
With these in hand, it suffices to show that $\Psi(G^{*},F_{k}^{*})\leq\Psi(G^{*},F_{k})$
in order to prove that $\Psi(G^{*},F^{*})\leq\Psi(F,G^{*})$.

\paragraph{Saddle point inequality in the $k$-bidder auction.}

Now we will show that $F_{k}^{*}$ is a best response among distributions
$F_{k}\in\mathcal{P}(V_{k},\{m_{i}\}_{i=1}^{k})$:

\[
\Psi(G^{*},F_{k}^{*})\leq\Psi(G^{*},F_{k})
\]

\begin{proof}
Recall that $F_{k}^{*}$ follows the ``L''-shape distribution for
the $k$ bidders. In other words, $\text{supp}(F_{k}^{*})$ consists
of value profiles ${\bf v}_{i}^{k}(x)$, in which $i$th bidder's
value is some $x\in[\alpha,1]$, and the values of the remaining $k-1$
bidders are set to $\alpha$.

We use Lemma \ref{lem:duality} to verify the optimality of $F_{k}^{*}$.
Recall that Lemma \ref{lem:duality} states that if there exists an
affine function $L(\mathbf{v})$ such that $L(\mathbf{v})\leq\phi(\mathbf{v};G^{*})$
for all ${\bf v}\in V$, and in addition, $L({\bf v})=\phi(\mathbf{v};G^{*})$
on the support of $F_{k}^{*}$, then $F_{k}^{*}$ is an optimal strategy
for nature. For our purpose, consider the following affine function:

\[
L(\mathbf{v})=\sum_{i=1}^{k}\frac{1}{k-1-\ln(\alpha)}v_{i}-\frac{\alpha}{k-1-\ln(\alpha)}.
\]

We will first show that $\phi(\mathbf{v};G^{*})=L(\mathbf{v})$ for
all $\mathbf{v}\in\text{supp}(F_{k}^{*})$. First, for any realization
$\mathbf{v}$, it is without loss to reorder the bidders' valuations
in ${\bf v}$ from highest to lowest $v_{1}\geq v_{2}\geq...\geq v_{k}$.
We can write:

\begin{align*}
\phi(\mathbf{v};G^{*}) & =\int_{0}^{1}\max\{v_{2},p\}\cdot1_{\{v_{1}>p\}}dG^{*}\\
 & =v_{2}G^{*}(v_{2})+\int_{v_{2}}^{v_{1}}rg^{*}(r)dr\\
 & =v_{2}G^{*}(v_{2})+\frac{v_{1}-v_{2}}{k-1-\ln(\alpha)}.
\end{align*}

Recall that the candidate strategy for the seller $G^{*}$ is:\footnote{While the current equilibrium puts mass of $\frac{k-1}{k-1-\ln(\alpha)}$
on $p=0$, there could be other equilibria in which the seller may
put density on $[0,\alpha]$. All of these equilibria are payoff equivalent
but may not yield the same behavior.}

\[
G^{*}(p)=\begin{cases}
\int_{\alpha}^{p}g^{*}(r)dr+\frac{k-1}{k-1-\ln(\alpha)} & \text{if }p>\alpha\\
\frac{k-1}{k-1-\ln(\alpha)} & \text{if }p\leq\alpha,
\end{cases}
\]
where

\[
g^{*}(r)=\begin{cases}
\frac{1}{r\left(k-1-\ln(\alpha)\right)} & \text{if }r>\alpha\\
0 & \text{if }r\leq\alpha.
\end{cases}
\]

Take any point ${\bf v}_{i}^{k}(x)$ in the ``L-shape'' support.
It is without loss to consider ${\bf v}_{1}^{k}(x)=(x,\alpha,\alpha,...,\alpha)$,
as the following holds symmetrically in all ${\bf v}_{i}^{k}(x)$:

\begin{align*}
\phi(x,\alpha,...,\alpha;G^{*}) & =\alpha G^{*}(\alpha)+\int_{\alpha}^{x}rg^{*}(r)dr\\
 & =\frac{(k-1)\alpha}{k-1-\ln(\alpha)}+\frac{x}{k-1-\ln(\alpha)}-\frac{\alpha}{k-1-\ln(\alpha)}\\
L(x,\alpha,...,\alpha) & =\frac{(k-1)\alpha}{k-1-\ln(\alpha)}+\frac{x}{k-1-\ln(\alpha))}-\frac{\alpha}{k-1-\ln(\alpha)}.
\end{align*}

Hence, $L$ and $\phi$ are identical (or intersect) on the support
of $F^{*}$.

Now we will show that for all $\mathbf{v}\in V$, $\phi(\mathbf{v};G^{*})\geq L(\mathbf{v})$.
We define the following functions in order to simplify the analysis:

\begin{align*}
\varphi(u,q) & :=\phi(u,q,q,...,q;G^{*}),\\
\ell(u,q) & :=L(u,q,q,...,q).
\end{align*}
$\varphi(u,q)$ and $\ell(u,q)$ coincide with $\phi$ and $L$ respectively
when the highest bidder's value is $u$ and the values of the other
bidders are equal to $q\leq u$.

Since $\phi$ only depends on the highest and second highest bidder,
the realization of $\phi$ will not be altered if the rest of the
valuations $v_{3},...,v_{n}$ were set to the second highest bidder's
value $v_{2}$. Hence,

\[
\phi(\mathbf{v})=\varphi(v_{1},v_{2}).
\]

Additionally, since $L$ is increasing in each coordinate,

\[
\ell(v_{1},v_{2})\geq L(\mathbf{v}).
\]

These functions allow us to compare $\phi$ and $L$ at various value
profiles without having to specify $v_{3},...,v_{k}$.

There are three separate cases to consider in order to show $\phi\geq L$:
(1) the second highest bidder's value is above the lower bound $\alpha$,
(2) the highest bidder's value is above $\alpha$ but the second highest's
is below $\alpha$, and (3) all values are below $\alpha$.

{\bf Case (1):} $v_{2}>\alpha$.

Observe first that for any $u>q$,

\[
\frac{\partial\varphi}{\partial v_{2}}(u,q)=G^{*}(q).
\]

For any ${\bf v}$ such that $v_{2}>\alpha$,

\begin{align*}
L({\bf v})\leq\ell(v_{1},v_{2}) & =\int_{0}^{v_{2}}\frac{\partial\ell(v_{1},q)}{\partial q}dq=\int_{0}^{v_{2}}\frac{k-1}{k-1-\ln(\alpha)}dq=\int_{0}^{v_{2}}G^{*}(0)dq\\
\phi(\mathbf{v};G^{*})=\varphi(v_{1},v_{2}) & =\int_{0}^{v_{2}}\frac{\partial\varphi(v_{1},q)}{\partial q}dq=\int_{0}^{v_{2}}G^{*}(q)dq,
\end{align*}
where $G^{*}(0)=\frac{k-1}{k-1-\ln(\alpha)}$ refers to the mass
of the atom at $p=0$. Since $v_{2}>\alpha$, $G^{*}(v_{2})\geq G^{*}(\alpha)=G^{*}(0)$
and therefore $\phi(\mathbf{v};G^{*})=\varphi(v_{1},v_{2})\geq\ell(v_{1},v_{2})\geq L(\mathbf{v})$.

{\bf Case (2):} $v_{1}\geq\alpha\geq v_{2}$.

For $v_{2}$ less than $\alpha$, $G^{*}(v_{2})=G^{*}(0)=\frac{k-1}{k-1-\ln(\alpha)}$,
since $G^{*}$ does not have any mass in $(0,\alpha)$. Hence,

\begin{align*}
L({\bf v})\leq\ell(v_{1},v_{2}) & =\frac{v_{1}}{k-1-\ln(\alpha)}+\frac{(k-1)v_{2}}{k-1-\ln(\alpha)}-\frac{\alpha}{k-1-\ln(\alpha)}\\
\phi({\bf v};G^{*})=\varphi(v_{1},v_{2}) & =v_{2}G^{*}(0)+\frac{(v_{1}-\alpha)}{k-1-\ln(\alpha)}\\
 & =\frac{v_{1}}{k-1-\ln(\alpha)}+\frac{(k-1)v_{2}}{k-1-\ln(\alpha)}-\frac{\alpha}{k-1-\ln(\alpha)}.
\end{align*}

Therefore, $\phi({\bf v};G^{*})\geq L({\bf v})$.

{\bf Case (3):} $\alpha>v_{1}$.

Since as mentioned above, $G^{*}$ only has mass at $p=0$ within
the interval $[0,\alpha]$, the expected revenue is $\phi({\bf v};G^{*})=v_{2}G^{*}(0)$.
In other words, the expected revenue is simply the second highest
bidder's value times the probability that a sale occurs, which is
the probability that $p<\alpha$, since $v_{1}<\alpha$. Note

\begin{align*}
L({\bf v})\leq\ell(v_{1},v_{2}) & =\frac{v_{1}}{k-1-\ln(\alpha)}+\frac{(k-1)v_{2}}{k-1-\ln(\alpha)}-\frac{\alpha}{k-1-\ln(\alpha)}\\
 & \leq\frac{(k-1)v_{2}}{k-1-\ln(\alpha)},\text{ and }\\
\phi({\bf v};G^{*})=\varphi(v_{1},v_{2}) & =v_{2}G^{*}(0)=\frac{(k-1)v_{2}}{k-1-\ln(\alpha)}.
\end{align*}

Therefore, $\phi({\bf v};G^{*})\geq L({\bf v})$ for all $\mathbf{v}\in V$.

By Lemma \ref{lem:duality}, we have proven that $F_{k}^{*}$ is an
best response for nature out of distributions $F_{k}\in\mathcal{P}(V_{k},\{m_{i}\}_{i=1}^{k})$.
Hence, we have proved the optimality of $F^{*}$:

\[
\Psi(G^{*},F^{*})\leq\Psi(G^{*},F_{k}^{*})\leq\Psi(G^{*},F_{k})\leq\Psi(G^{*},F).
\]
\end{proof}

\subsection{\label{sec:Appendix:proof_fpa}Appendix: Proof of Lemma \ref{fpa-competitive}}
\begin{proof}
Without loss, fix any $F$ such that there exists a bidder say $i$
who may realize a clearly low value, say $v_{i}$. Suppose that bidder
$i$ indeed has realized the clearly low value $v_{i}$. Then, by
definition, there exists another bidder $j$ who has a strictly higher
value than $v_{i}$ with probability one.

First, consider a second-price auction. The result follows trivially
since bidders participate if and only if the realized reserve price
is weakly less than their values, and upon participation, they adopt
weak dominant strategies of bidding their values. Hence, regardless
of the realized reserve price, bidder $i$ never wins whenever she
has a clearly low value, since there is another bidder with strictly
higher value with probability one.

Next, consider a first-price auction. Given $I_{0}$, it is common
knowledge that $i$ knows that there exists a bidder $j$ who has
a strictly higher value than $v_{i}$ with probability one. and bidder
$j$ knows that $i$ has value $v_{i}$. It suffices to show that
bidder $i$ never wins under $I_{0}$. Suppose the (realized) reserve
price $r$ is strictly higher than $v_{i}$. Then, bidding above $r$
is weakly dominated. Hence, bidder $i$ must win with zero probability,
so we are done.

Hence, assume the (realized) reserve price $r$ is weakly less than
$v_{i}$. Let $\overline{b}_{i}$ and $\underline{b}_{i}$ respectively
denote the supremum and infimum of the support of bidder $i$'s equilibrium
bid. The dominance restriction means that $\overline{b}_{i}\le v_{i}$.
Next, let $\overline{b}$ and $\underline{b}$ respectively denote
the supremum and infimum of the support of $B_{-i}:=\max_{k\ne i}B_{k}$,
where $B_{k}$ is the random variable representing bidder $k$'s equilibrium
bid. (We let the equilibrium bid be 0 when a bidder does not participate.)

We next prove that $\underline{b}\ge v_{i}$. Suppose to the contrary
that $\underline{b}<v_{i}$. Then, bidder $i$ must enjoy strictly
positive surplus in equilibrium. This means either $\underline{b}_{i}>\underline{b}$,
or $\underline{b}_{i}=\underline{b}$ \textit{and} $B_{-i}$ has a
point mass at $\underline{b}$. (Otherwise, if bidder $i$ must sometimes
make a bid in equilibrium that loses with probability arbitrarily
close to one, which is inconsistent with him earning strictly positive
surplus.) We derive a contradiction in each case.
\begin{enumerate}
\item Suppose $\underline{b}_{i}>\underline{b}$: Then, there must exist
$\hat{b}\in(\underline{b},\underline{b}_{i})$ such that bidder $j$
must bid $B_{j}<\hat{b}$ with strictly positive probability, and
when he does, he loses with probability one. Note, for each type $v_{j}>v_{i}$
of bidder $j$, there is $\epsilon>0$ such that $\underline{b}_{i}+\epsilon<v_{j}$,
and for that type, deviating to any bid $\underline{b}_{i}+\epsilon$
would win with positive probability and positive surplus upon winning,
so there is a profitable deviation for $j$.
\item Suppose $\underline{b}_{i}=\underline{b}$ and $B_{-i}$ has a point
mass at $\underline{b}$: In this case, there are two possibilities.
Suppose first either bidder $i$ loads zero mass at $\underline{b}_{i}$
or bidder $j$ bids strictly below $\underline{b}_{i}$ with positive
probability. Either case, bidder $j$ is making a bid in the putative
equilibrium that would surely lose. One can then construct a profitable
deviation analogous to the case of (i) for bidder $j$. Suppose next
both bidder $i$ and bidder $j$ both load positive mass at $\underline{b}_{i}=\underline{b}$.
In this case, raising his bid above $\underline{b}$ but arbitrarily
close to $\underline{b}$ results in a discrete increase in winning
probability (and virtually the same positive surplus for all type
$v_{j}>v_{i}$) for bidder $j$, and thus constitutes a profitable
deviation for $j$.
\end{enumerate}
The contradictions thus obtained imply that $\underline{b}\ge v_{i}$.

We next claim that bidder $i$ never wins with positive probability.
Since $\underline{b}\ge v_{i}\ge\overline{b}_{i}$, this is possible
only when bidder $i$ bids $v_{i}$ with positive probability and
$B_{-i}$ has a point mass at $\underline{b}=v_{i}$. However, the
argument used in (ii) then produces a profitable deviation for bidder
$j$ and yields a contradiction. 
\end{proof}

\subsection{\label{sec:Appendix:proof_theorem_2}Appendix: Proof of Theorem \ref{thm:optimal_mechanism}}

\begin{proof}
We have already proven that the inequality $\mathcal{R}\left(F^{*},I_{0},S_{G^{*}}\right)\leq\mathcal{R}\left(F,I_{0},S_{G^{*}}\right)=\mathcal{R}\left(F,I,S_{G^{*}}\right)$
for all $F\in\mathcal{P}(V,\{m_{i}\}_{i=1}^{n})$ and $I\in\mathcal{I}$
in the proof of Theorem \ref{thm:saddle_point}. It suffices to prove
that $\mathcal{R}\left(F^{*},I_{0},M\right)\leq\mathcal{R}\left(F^{*},I_{0},S_{G^{*}}\right)$
for all $M\in\widehat{\mathcal{M}}$.

Take any mechanism $M\in\widehat{\mathcal{M}}$. Recall that under
the UCP information structure $I_{0}$, any mechanism $M$
has a representation as a direct mechanism by the revelation principle,
and overloading notation, we refer to this direct mechanism as $M=(V,x,t)$.
If $M$ does not induce a Bayes-Nash equilibrium in undominated strategies,
then the revenue is set to zero, and it will be strictly dominated
by the second-price auction. If $B(M,I_{0},F^{*})$ is non-empty,
then $M$ must satisfy the Bayesian individual-rationality (\ref{eq:IR-0})
and incentive-compatibility (\ref{eq:IC-0}) stated in Section \ref{sec:Model}.

We proceed to evaluate the revenue obtained by $M$ under the value
distribution $F^{*}$. First, the $n-k$ inactive bidders in $F^{*}$
will have clearly low values almost surely, since they all have a
point mass support under $F^{*}$ on their means, which are strictly
less than $\alpha$. They are never allocated the good under any competitive
mechanism $M$. Hence, we can focus solely on the revenue from the
$k$ active bidders.

Following the standard mechanism design algebra, the expected revenue
from any active bidder $i$ is as follows:

\[
\mathbb{E}_{\left(F^{*},I_{0}\right)}[t_{i}(\mathbf{v})]=E_{v_{i}}[T_{i}(v)]=\int_{\alpha}^{1}[vX_{i}(v)-U_{i}(v)]dF_{i}^{*}(v),
\]
where $F_{i}^{*}$ is the marginal distribution for bidder $i$ under
$F^{*}$, $T_{i}(v)$ is the interim payment by bidder $i$, and $X_{i}(v)$
is the interim allocation, all when bidder $i$'s valuation is $v$.

Recall that under $I_{0}$, bidders have common knowledge of the distribution
$F^{*}$ but do not gain any additional information of their opponents'
values from their signals. Hence, the interim values are solely determined
by averaging over the opponents' valuations according to $F^{*}$.
Of course, since $F^{*}$ does not exhibit independent marginal distributions
for each bidder's valuation, the interim values will depend on bidder
$i$'s realized type.

Recall that each active bidder in $F^{*}$ is chosen as the highest
bidder with probability $\theta_{i}$, with valuation distributed
according to $H$, and is a second-highest bidder with probability
$1-\theta_{i}$, with a valuation set at $\alpha$. We may write

\begin{align*}
\int_{\alpha}^{1}[vX_{i}(v)-U_{i}(v)]dF_{i}^{*}(v) & =\theta_{i}\left(\int_{\alpha}^{1}[vX_{i}(v)-U_{i}(v)]dH(v)\right)+(1-\theta_{i})\left(\alpha X_{i}(\alpha)-U_{i}(\alpha)\right)\\
 & =\theta_{i}\left(\alpha\left(X_{i}(1)-U_{i}(1)\right)+\int_{\alpha}^{1}\left[vX_{i}(v)-U_{i}(v)\right]h(v)dv\right)\\
 & +(1-\theta_{i})\left(\alpha X_{i}(\alpha)-U_{i}(\alpha)\right).
\end{align*}

Note that any active bidder with valuation equal to $\alpha$ is a
clearly-low bidder. Conditional on bidder $i$ having a valuation
of $\alpha$, some bidder $j\neq i$ must have been selected as the
highest bidder, whose valuation exceeds $\alpha$ with probability
1. Moreover, the highest bidder knows that all other active bidders
have value $\alpha$. Hence, for any competitive mechanism $M\in\widehat{\mathcal{M}}$,
$X_{i}(\alpha)=0$. This will also imply that $U_{i}(\alpha)=0$,
since such a bidder is never allocated the good and never enjoys strictly
positive surplus in equilibrium.

By the envelope theorem \citep{borgers}, $U_{i}(v)=\int_{\alpha}^{v}X_{i}(s)ds+U_{i}(\alpha)=\int_{\alpha}^{v}X_{i}(s)ds$.
We can then simplify the expression:

\begin{align*}
\int_{\alpha}^{1}\left[vX_{i}(v)-U_{i}(v)\right]h(v)dv & =\int_{\alpha}^{1}vX_{i}(v)h(v)dv-\int_{\alpha}^{1}\left(\int_{\alpha}^{v}X_{i}(s)ds\right)h(v)dv\cdot\\
 & =\int_{\alpha}^{1}vX_{i}(v)h(v)dv-\int_{\alpha}^{1}\left(\int_{v}^{1}h(s)ds\right)X_{i}(v)dv\\
 & =\int_{\alpha}^{1}\left(v-\frac{H(1_{-})-H(v)}{h(v)}\right)X_{i}(v)h(v)dv_{i}\\
 & =\int_{\alpha}^{1}\left[J(v)X_{i}(v)\right]h(v)dv
\end{align*}
where $H(1_{-})$ is the left limit of $H$ at $1$ and $J(v)$ is
the virtual value:

\[
J(v)\equiv v-\frac{H(1_{-})-H(v)}{h(v)}=v-\frac{(1-\alpha)-\frac{v-\alpha}{v}}{\frac{\alpha}{v^{2}}}=v+v^{2}-v=v^{2}.
\]

Plugging this expression back into the expected revenue, we obtain:

\begin{align*}
\mathbb{E}_{\left(F^{*},I_{0}\right)}[t_{i}(\mathbf{v})] & =\theta_{i}\left(\alpha\left(X_{i}(1)-\int_{\alpha}^{1}X_{i}(s)ds\right)+\int_{\alpha}^{1}v^{2}X_{i}(v)h(v)dv\right)+(1-\theta_{i})\left(\alpha X_{i}(\alpha)\right)\\
 & =\theta_{i}\left(\alpha X_{i}(1)-\alpha\int_{\alpha}^{1}X_{i}(s)ds+\alpha\int_{\alpha}^{1}X_{i}(v)dv\right)+(1-\theta_{i})\left(\alpha X_{i}(\alpha)\right)\\
 & =\theta_{i}\left(\alpha X_{i}(1)\right)+(1-\theta_{i})\left(\alpha X_{i}(\alpha)\right)\\
 & =\theta_{i}\left(\alpha X_{i}(1)\right).
\end{align*}

 We can now sum up the revenue from all active bidders:

\[
\mathcal{R}(F^{*},I_{0},M)=\sum_{i=1}^{k}\mathbb{E}_{\left(F^{*},I_{0}\right)}[t_{i}(\mathbf{v})]=\alpha\sum_{i=1}^{k}\theta_{i}X_{i}(1)\leq\alpha=\mathcal{R}(F^{*},I_{0},S_{G^{*}}),
\]
since $X_{i}(1)\leq1,\forall i$ and $\sum_{i=1}^{k}\theta_{i}=1$.

Hence, the optimal second-price auction obtains the optimal revenue
guarantee out of all mechanisms in $\widehat{\mathcal{M}}$.
\end{proof}

\end{document}